\documentclass[a4paper,12pt]{article}

\usepackage[T1]{fontenc}
\usepackage[english]{babel}
\usepackage[latin1]{inputenc}

\usepackage{amsfonts}
\usepackage{amsmath} 
\usepackage{amssymb}
\usepackage{amsthm}
\usepackage{appendix}
\usepackage{caption}
\usepackage{cases}

\usepackage{manfnt}
\usepackage{mathrsfs}
\usepackage{minitoc}
\usepackage[numbers]{natbib}

\usepackage{pgf}
\usepackage{pifont}
\usepackage{pstricks}

\usepackage{graphics}
\usepackage{graphicx}

\usepackage{lastpage}
\usepackage{latexsym}

\usepackage{color}
\usepackage{comment}

\usepackage{empheq}

\usepackage{enumerate}

\usepackage{fancyhdr}
\usepackage[top=2.5cm, bottom=3cm, left=2cm , right=2.1cm]{geometry}
\newenvironment{pf}[1][Proof]{\noindent\textsc{#1.} }{\ \rule{0.5em}{0.5em}}

\theoremstyle{plain}
\newtheorem{Theorem}{Theorem}[section]
\newtheorem{Corollary}[Theorem]{Corollary}

\newtheorem{Definition}[Theorem]{Definition}
\newtheorem{Proposition}[Theorem]{Proposition}
\newtheorem{Remark}{Remark}[section]

\numberwithin{equation}{section}

\def\demi{\frac{1}{2}}

\def\cal{\mathcal}

\def\L{{\cal L}}

\def\B{{\mathcal B}}

\def\F{{\mathcal F}}
\def\K{{\cal K}}

\def\sigR{{\cal R}}
\def\sigoR{{\cal R}^{\perp}}
\def\sR{\text{{\tiny${\cal R}$}}}
\def\S{{\cal S}}

\def\U{{\cal U}}

\def\bfF{{\bf F}}

\def\bfU{{\bf U}}
\def\bfV{{\bf V}}

\def\E{{\mathbb E}}
\def\bF{{\mathbb F}}

\def\P{{\mathbb P}}
\def\Q{{\mathbb Q}}
\def\R{{\mathbb R}}

\def\R{{\mathbb R}}

\def\tU{{\widetilde U}}
\def\tu{{\tilde u}}
\def\tv{{\tilde v}}
\def\tmu{{\tilde \mu^*}}
\def\tsigma{{\tilde \sigma^*}}
\def\wL^*{{\widehat L^{(\mu^*, \sigma^*)}}}
\def\wL{{\widehat L}}

\def\tV{{\widetilde V}}

\def\Cc{{\mathcal C}}

\def\GX{{\mathscr X}}
\def\GY{{\mathscr Y}}

\def\eps{\varepsilon}

\def\cgblue{\color{blue}}
\def\cbred{ \color{red}}

\newcommand{\rmi}{{\rm (i) $\>\>$}}

\newcommand{\rmii}{{\rm (ii) $\hspace{1.5mm}$}}
\newcommand{\rmiii}{{\rm (iii)$\>\>$}}
\newcommand{\rmiv}{{\rm (iv)$\>\>$}}

\newcommand{\rma}{{\rm a)$\>\>$}}
\newcommand{\rmb}{{\rm b)$\>\>$}}

\def\B{\Big}
\def\b{\big}

\def\bit{\begin{itemize}}
\def\eit{\end{itemize}}

\def\bc{\begin{center}}
\def\ec{\end{center}}

\def\super { \end{document}}
\def\bcom{}
\def\edoc{\end{document}}

\DeclareMathOperator{\esssup}{ess\,sup}

\title{ Ramsey Rule with Progressive   Utility \\
 in  Long Term Yield Curves Modeling\thanks{With the financial
support of the "Chaire Risque Financier of the  Fondation du Risque" } }

\author{ El Karoui Nicole,
\thanks{ \small  LPMA, UMR CNRS  6632,  Universit\'e Pierre et Marie Curie, CMAP, UMR CNRS 7641, \'Ecole Polytechnique  }
\\
\and Hillairet Caroline,~ \thanks
{\small CMAP, UMR CNRS 7641, \'Ecole Polytechnique,}
\and Mrad~Mohamed~ \thanks
{\small   LAGA, UMR CNRS 7539,  Universit\'e Paris 13}
}
\date{\today}

\begin{document}
 \maketitle

  \abstract{The  purpose of this paper relies on  the study of long term yield curves modeling.  Inspired by the economic litterature, it provides a   
  financial interpretation of the Ramsey rule that links discount rate and marginal utility of aggregate optimal consumption. For such a  long maturity modelization, the possibility of adjusting
preferences to new economic information is crucial. Thus, after recalling some important properties on progressive utility, this paper first
 provides an extension of the notion of a consistent  progressive utility  to a consistent pair of
progressive  utilities of investment and consumption. An optimality condition is  that the utility from the wealth satisfies  a    second order SPDE of HJB type involving the Fenchel-Legendre transform of the utility from consumption. This SPDE is solved in order to give a full characterization of this class of consistent progressive  pair of utilities. An application of this results  is to   revisit the classical backward optimization problem in the light of progressive utility theory, 
 emphasizing  intertemporal-consistency issue.  Then we 
study the dynamics of the marginal utility yield curve, and give example with backward and progressive power utilities.
     }\\



\section*{Introduction}

This paper focuses on the modelization of long term yield curves. For the financing of ecological project, for the pricing of longevity-linked securities or any other investment with long term impact, modeling long term interest rates is crucial. The answer cannot be find in financial market since  for longer maturities, the bond market is highly illiquid and standard financial interest rates models cannot be easily extended. Nevertheless, an abundant literature
on the economic aspects of long-term policy-making (i.e. a time horizon between 50 to
200 years), has been developed. The issue is adressed at a macroeconomic level, where long run interest rates have not necessarily the same meaning than in financial market. 
The Ramsey rule, introduced by Ramsey in his seminal work \cite{Ramsey} and further discussed by numerous economists such as Gollier \cite{Gollier6,Gollier9,Gollier13,Gollier14,Gollier15,Gollier16,GollierEcological} and Weitzman \cite{Weitzman,
Weitzman_review}, is the reference equation to compute discount rate, that allows to evaluate the future value of an investment by giving a current equivalent value. The Ramsey rule links the discount rate with the marginal utility of aggregate consumption at the economic equilibrium. Even if this rule is very simple, there is no consensus among economists about the parameters that should be considered, leading to very different discount rates.  But economists agree on the necessity of a sequential decision scheme that allows to revise the first decisions in the light of new  knowledge  and  direct experiences:  the utility criterion must be adaptative and adjusted to the information flow. In the classical optimization point of view, this adaptative criteria is called  consistency. 
In that sense, market-consistent progressive utilities,  studied in El Karoui and Mrad
\cite{MrNek01,MrNek02},  are the appropriate tools to study long term yield curves.

 Indeed, in  a dynamic and stochastic environment, the classical notion of utility function  is not
flexible enough  to help us to make good choices in the long run. M. Musiela and T. Zariphopoulou 
\cite{zar-07a,zar-08,zar-07}   were
the first to suggest to use instead of the classical criterion  the concept of
 progressive dynamic utility, that gives an adaptative way to  model possible changes over the time of individual preferences of an agent.   Obviouslly the dynamic utility must be  
 consistent with respect to a given investment universe; this question has been studied from a PDE point of view in  \cite{MrNek01}.
 Motivated by the Ramsey rule (in which the consumption rate is a key process), we extend the notion of market-consistent progressive  utility  
 with consumption: the agent invest in a financial market and consumes a part of her wealth at each instant. This progressive utilities of investment 
 and consumption were considered at first by Berrier and Tehranchi \cite{Mike02} in the  particular case of a zero volatility. This paper studies the general case with a different approach.

In a financial framework, it is natural to link yield curves and zero-coupon, whose pricing in incomplete market is a complex question. 
Utility functions are also the cornerstone in the utility indifference pricing method, for the pricing of non-replicable contingent claim. 
For a small amount of transaction, this pricing method  leads to a linear pricing rule (see \cite{Davis})  called the  {Davis price} or {marginal utility price}. 
As the zero-coupon bond market is highly illiquid for long maturity,  it is relevant to study  utility indifference pricing method for progressive
utility with consumption.  This paper also points out the similarities and the differences between progressive utilities and the value function of backward classical 
utility maximization problem. Although the backward classical value function is a progressive utility  (cf Mania and Tevzadze 
\cite{Mania} for the case without consumption), the way the classical optimization problem is
posed is very different from the progressive utility problem. In the classical approach,  the optimal  processes are   computed through a backward analysis, 
emphasizing their dependency  on the horizon  of the optimization problem, and leading to intertemporality issues. In the progressive approach, 
we propose regularity conditions on the utilities characteristics that ensure the existence of consistent utilities and of optimal solutions. 

 We illustrate those issues on the example of long term discount rate and yield curves. According to the Ramsey rule, we show that equilibrium 
interest rate and marginal utility interest rate coincide, being careful that this last curve is robust only for small trades. For replicable bonds,
equilibrium interest rate and  market interest rate are the same. Finally, we study the dynamics of the marginal utility yield curve, in the framework of
progressive and backward power utilities (since power utilities are the most commonly used in the economic literature).  Special attention is paid on the impact on the yield curves of the
maturity of the underlying optimization problem.

The paper is organized as follows, with a concern for finding a workable accommodation between intuition and technical results. For more technical details, 
the interested reader may refer to \cite{MrNek01}. 
Section 2 starts with  the definition
of  It\^o progressive   utilities and characterizes these concave
It\^o's random fields as primitives of  SDEs. A special attention is paid to the dynamics of the Fenchel conjuguate
utility random field, yielding to a very intuitive SPE for the marginal conjuguate utility.  Section 3 is a technical section where as in H.Kunita \cite{Kunita:01}, 
"Sobolev spaces" of processes are introduced, in order to study rigorously
the properties of monotonicity, differentiability and concavity, both for random
fields  and   solutions of SDEs. Then, the link between non linear SPDE and SDE is detailed, providing a path representation of solution of SDEs.

Section 4 introduces the investment universe and studies  market-consistent
progressive utilities of investment and consumption.  From consistency property we derive a SPDE of HJB type  satisfied by the dynamic utility of investment and consumption.  
Based on the  connection between SDEs and SPDEs developed in Section $3$  and using same  stochastic flows technics as  in \cite{MrNek01}, a closed formula 
for these forward consistency utilities is given, in term of the inverse flow of the optimal wealth. 
Special attention is paid to  the example of power consistent utility  This section ends with some results on
 marginal utility indifference pricing, as an application of utility maximization.   
 
 Application to yield curve dynamics is given in Section 5. 
After introducing the economic framework for the computation of long term discount rates, we give a financial interpretation of the Ramsey rule 
and we study the dynamics of the  marginal utility yield curve. More precise properties of the yield curve are given in the framework of power utilities and log-normal market, in particular on the impact of the terminal horizon.

\section{Progressive Utility}
 Motivated by the necessity of more flexible criterium with respect to the uncertainty of the universe, we introduce the notion of progressive utility.
All stochastic processes are defined on a standard filtered probability space
$(\Omega,{\mathbb F},\mathbb{P})$, where
the filtration ${\bF}=(\F_t)_{t\ge 0}$ is assumed to be right continuous and
complete. The probability measure $\mathbb{P}$ is a reference probability, often the historical probability.
\paragraph{Progressive utility and its Fenchel conjugate}
 We start with the definition of a progressive utility as progressive random
field on $\R^*_+$ concave and increasing with respect to the parameter. Given its importance in convex analysis, we introduce together 
its convex conjugate $\bf \widetilde U$ (also called conjugate progressive  utility (CPU)).

\begin{Definition}[{Progressive Utility}]\label{defPUF}${}$\\
\rmi \rma A {\em progressive utility} is a continuous progressive random field on $\R^*_+$,
$\bfU=\{U(t, x); t \geq 0, x > 0\}$ such that, for every $(t,\omega)$, $x\mapsto U(\omega, t, x)$
is a strictly concave, strictly increasing, and 
non negative utility function.\\
 \rmb {\sc Inada Condition} $\bf U$ is assumed to be ${\mathcal C}^2$-random field,  satisfying Inada conditions:
 for every $(t,\omega)$, $U(t,\omega,x)$ goes to $0$ when $x$ goes to $0$ and the
derivative $U_x(t,\omega,x)$ (also called marginal utility) decreases from $\infty$  to $0$.\\
\rmii  The progressive {\em convex conjugate} (also called Fenchel conjuguate) of the progressive utility $\bf U$ is
the progressive random field $\bf \tU$ defined on  $\R^*_+$ by 
\bc
${\bf \tU}=\{\tU(t,y); t\geq 0,y> 0\}$, where
$\tU(t,y)\stackrel{def}{=}\max_{x>0,x\in Q^+}\big(U(t,x)-x\, y\big)$.
\ec
Under  Inada condition, ${\bf \tU}$ is twice continuously
differentiable, strictly convex, strictly decreasing, with 
$\tU(.,0^+)=U(+\infty), \>
\tU(.,+\infty)=U(0^+), a.s.$\\
\rmiii The marginal  utility random field $\bfU_x$  is the inverse of the opposite
of the   marginal conjugate utility random field ${\bf \tU}_y$, that is
$U_x(t,.)^{-1}(y)=-\tU_y(t,y)$, with  $\tU_y(.,0^+)=-\infty$, 
$\tU_y(.,+\infty)=0$, under Inada condition.\\
\rmiv The bi-dual relation holds true $U(t,x)=\inf_{y>0,y\in
Q^+}\big(\tU(t,y)+x\, y\big)$. \\Moreover $\tU(t,y)=U\big(t,-\tU_y(t,y)\big) +\tU_y(t,y)\, y$,
and $U(t,x)=\tU\big(t, U_x(t,x)\big)+x\, U_x(t,x)$. 
\end{Definition}
Progressive utility is an example of stochastic process depending on a real parameter $x$, also called {\em progressive random
field} $\bf X$. It is useful to specify in some sense some properties that have to be  considered when this additional
parameter $x$ is taken into account. In particular, we say that the random field ${\bf X}\in  \F_{\infty}\otimes {\mathcal
B}(\R^+)\otimes {\mathcal B}(\R^+)$ satisfies a property  $\cal P$, if there exists $N\in \F_{\infty}$ with $\P(N)=0$, such
that the property is satisfied on $N^c$. For instance, a random field $X$ is said to be progressive, (predictable, optional)
if there exists  $N\in \F_{\infty}$ such that for every $\omega \in N^c$, for every $x \in \R+$, the process $t \mapsto
X_t(\omega, x)$ is progressively measurable. Another family of examples is given by properties relative to the parameter $x$:
for any $\omega\in N^c$, for every $t>0$, $x\mapsto X(t,x)(\omega)$ satisfies the
property $\cal P$. In particular, all previous properties as concavity, derivability and so on, may be understand in this
sense. The symbol $\P \>\>a.s. $ is used to said that the negligeable set is not depending on $x$. 

To highlight the intuition, Section $2$ presents the key ideas that will guide us throughout the rest of this work,  with little regard to the assumptions. Section  $3$ completes then the study by focusing on the conditions under which our assumptions are  satisfied.
 \section{It\^o's Progressive  Utility}\label{SPU}
This section uses tools developed in \cite{MrNek01}  and recalls some important results  on It\^o's progressive  utility that will be useful for  this work.
\subsection{It\^o 's progressive utility and SDE}\label{ISPU}
We focus  on  continuous progressive utilities $\bfU$ which are 
a collection of  It\^o's  semimartingales depending on a parameter driven by a  $n$-dimensional Brownian
motion $W=(W^1,..,W^n)$ defined on the probability space $(\Omega,{\mathbb
F},\mathbb{P})$. 
From H.Kunita \cite{Kunita:01},  there exist two progressive random fields $( \beta(t,x),
\gamma(t,x))$,  called {\em local characteristics} of $\bfU$
 so that $\P- a.s.,$
 \begin{equation}\label{dynbetagamma}
 dU(t,x)=\beta(t,x)dt + \gamma(t,x).dW_t
\end{equation}
As usual, the random field $\beta$ is called the drift characteristic, and the random field $\gamma$ is called the diffusion characteristic.  For $t=0$, the deterministic utilities $U(0,.)$ and $V(0,.)$ are denoted $u(.)$ and $v(.)$ and in the following small letters $u$ and $v$ design deterministic utilities while capital letters refer to progressive utilities.  \\
  A first step  is to give conditions 
   on the local characteristics $(\beta, \gamma)$ such that the progressive random field $\bfU$
defined by \eqref{dynbetagamma} is a progressive utility, that is  
 monotonic and concave with respect to $x$. It is often easier to prove that  the progressive marginal utility $\bfU_x$
is strictly decreasing and strictly positive, with range $(0,\infty)$.

\begin{Proposition}\label{UxSDE} 
\rmi We assume ${\bf U}$ is regular enough, so that the first and second derivative random fields  $\bfU_x$ and $\bfU_{xx}$
are also  It\^o's  random fields, with local characteristics
$(\beta_x, \gamma_x)$, and $(\beta_{xx}, \gamma_{xx})$. We recall that $-\bfU_x$  is equal to the derivative of the conjugate utility $\tU_y $. \\
\rmii {\sc Intrinsic SDE} The marginal stochastic utility $\bfU_x$ (up to the change of initial condition  $x=-\tu_y(z)$) is a
strong solution  $Z_.(z)=U_x(.,-\tu_y(z))$ of the following one dimensional
stochastic differential equation SDE$(\mu, \sigma)$,  that is $\P \>\>a.s.,$
\begin{equation}\label{eq:ItoSDE}
\left\{
\begin{array}{cllll}
 dZ_t&=\mu(t,Z_t)dt+\sigma(t,Z_t)\,dW_t, & \quad Z_0 = z\\
\mu(t,z)&:=\beta_x\big(t,-\tU_y(t,z)\big),& \quad\mu(t,0)=0\\
\sigma(t,z)&:=\gamma_x\big(t,-\tU_y(t,z)\big),&\quad \sigma(t,0)=0
  \end{array}
 \right .
\end{equation}
The solution $Z$ is monotonic with respect to its initial condition, with range $(0,\infty)$.\\
\rmiii {\sc Stochastic utility characterization as primitive of SDE}  Let consider  
a SDE$(\hat{\mu},\hat{\sigma})$, $dZ_t=\hat{\mu}(t,Z_t)dt+\hat{\sigma}(t,Z_t)\,dW_t,  \> Z_0 = z$ and assume the existence of
 a strong global solution $Z_.(z)$,  increasing and differentiable
 in $z$ with range $(0,\infty)$. 
Then, for any utility function $u$ such that $Z_.(u_x(x))$ is Lebesgue-integrable in a
neighborhood of $x=0$, 
the primitive $\bfU=\{U(t,x)=\int_0^xZ_t(u_x(z))dz, t \geq 0, x>0\}$ is a
progressive utility.
\end{Proposition} 
\noindent {\bf Comment} \rmi The {\bf existence} of strong global solution of  SDE$(\mu, \sigma)$ is proved by using the same argument than in the
deterministic case, when the coefficients are uniformly Lipschitz, with (random) time depending Lipschitz bound, (Protter
\cite{Protter}, or for more exhaustive study, see Kunita \cite{Kunita:01}).  A constant Lipschitz bound $C$ corresponds to
the
classical framework of Lipschitz SDE,  and the range property is well-known. \\
\rmii The notion of "global solution" expresses that the solution $ (Z_t(z) )$  exists for all $ t \ge0
$. Under weaker assumptions, the solution may be defined  only up to a finite lifetime $ \zeta(z)$, before exploding. More
details will be
given in the next section.\\
\rmiii {\bf Sufficient conditions} on local characteristics $(\beta, \gamma)$ of an  It\^o's random
field $\bfU$ to be a progressive utility may be exhibited: in particular, 
 if there exist random Lipschitz bounds $C^i_t$ and $K^i_t$ with
$\int_0^TC^i_tdt<+\infty$ and  $\int_0^T|K^{i}_t|^{2}\,dt<+\infty$ for any $T$, such that
$\P \>a.s.$,
\begin{equation}
\left\{
\begin{array}{clll}
|\beta_x(t,x)\leq C_t \,|U_x(t,x)|, &        \|\gamma_x(t,x)\|\leq
K_t\,|U_x(t,x)|             \\
   |\beta_{xx}(t,x)|\leq C^1_t \,|U_{xx}(t,x)|, &\|\gamma_{xx}(t,x)\|\leq K^1_t
\,|U_{xx}(t,x)|\\
 \end{array}
 \right.
 \end{equation}
 The coefficients of the intrinsic SDE$(\mu,\sigma)$ are uniformly Lipschitz and $\bf U$ is a progressive utility.

\subsection{Dynamics of  Convex Conjugate Progressive Utility}\label{DynConjPU}
 The study of the convex conjugate $\bf \tU$ of a progressive utility 
 $\bf U$ is based on the well-known identity (Definition \ref{defPUF})
$\tU(t,y)=U(t,-\tU_y(t,y))+y \tU_y(t,y)$, and request to know the dynamics  of the ${\mathcal C}^2$-semimartingale
$ U (t, x) $ along the process $ - \tU_y (t, y) $. Calculations are based on  It\^o-Ventzel's formula, an
extension of the classical It\^o formula. We refer to Ventzel \cite{Ventzel} and Kunita \cite{Kunita:01} (Theorem $3.3.1$)  for different variants of this
formula. 
\begin{Proposition}[It\^o-Ventzel's Formula]\label{IVF}
Consider a ${\mathcal C}^2$-It\^o semimartingale $\bfF$ with local characteristics $(\phi,\psi)$, such that $\bfF_x$  is also
an It\^o semimartingale, with characteristics $(\phi_x,\psi_x)$.
 For any continuous  It\^o semimartingale $X$, 
$F(.,X_.)$ is an It\^o semimartingale, 
\begin{eqnarray}\label{ItoGen}
&&F(t,X_t)=F(0,X_0)+\int_0^t \phi(s,X_s)ds+\int_0^t\psi(s,X_s).dW_s \\
&+&\int_0^t F_x(s,X_s)dX_s +\frac{1}{2}\int_0^t F_{xx}(s,X_s)\langle dX_s\rangle
+ 
\int_0^t   \langle dF_x(s,x), dX_s \rangle|_{x=X_s} \nonumber
\end{eqnarray}
\end{Proposition}
\vspace{2mm}

\noindent
{\bf Comment}  The first line of the right hand side of the equation corresponds
to the dynamics of the process $(F(t,x))_{t\ge 0}$ taken on $(X_t)_{t\ge 0}$, 
when in the second line, the first two terms come from the classical It\^o's
formula. The last term  represents the quadratic covariation between $dF_x(t,x)$
and $dX_t$, at $x=X_t$, which can be written  as $\psi_x(t,X_t).\sigma^X_t dt$
when the diffusion coefficient of $X$ is the vector $\sigma^X_t$.\\
It\^o-Ventzel's formula and monotonic change of variable  will help us to establish
the relationship between  local characteristics of the random fields $\bfU$ and $\bf \tilde U$.
\begin{Theorem}\label{ResPrA}
Let $\bfU$  a progressive utility  and $\bf \tU$ its progressive convex conjugate
utility assumed to be ${\mathcal C}^2 $-It\^o's  semimartingales with local characteristics
$(\beta, \gamma)$ and $(\tilde\beta, \tilde\gamma)$. We also assume that  their marginal utilities $\bfU_x$, and $\bf \tU_y$ 
are It\^o's semimartingales with local characteristics $(\beta_x, \gamma_x)$ and $(\tilde\beta_y, \tilde\gamma_y)$.\\
\rmi The dynamics of  $\bf \tU$ is driven by the non linear second order 
SPDE,
\begin{equation}\label{eq:conjugate}
d\tU(t,y)=\gamma(t,-\tU_y(t,y)).dW_t+\beta(t,-\tU_y(t,y))dt +\frac{1}{2}
\tU_{yy}(t,y)\|\gamma_x\big(t,-\tU_y(t,y)\big)\|^2\,dt.
\end{equation}
\rmii Assume $(\mu, \sigma)$ (the random coefficients of the SDE associated with
${\bf U}_x$) to be fairly regular for 
 the adjoint elliptic operator in divergence form is well defined,
\begin{equation}\label{eq:elliptic}
\widehat L_{t,y}^{\sigma,\mu}(f)=\demi\partial_y(\|\sigma(t,y)\|^2
\partial_yf(t,y))-\mu(t,y) \partial_y f(t,y).
\end{equation}
Then the marginal conjugate utility $\bf \tU_y$ is a monotonic solution of
the  forward SPDE
\begin{eqnarray}\label{AdjointSPDEA}
d\tU_y(t,y)= -\partial_y(\tU_y)(t,y)\sigma(t,y).dW_t+\wL_{t,y}^{\sigma,\mu}(\tU_y) dt, \quad \tU_y(0,y)=\tu_y(y).
\end{eqnarray}
{\rm Observe that the derivability of the local characteristics 
$(\tilde\beta, \tilde\gamma)$ of $\bf \tU$ requires the existence of a third
 derivative for $\bf \tU$, and thus for $\bf U$. Remark also that $(ii)$ characterizes the inverse of a SDE.}
\end{Theorem}
\begin{pf} 
 Let   apply  It\^o-Ventzel's formula to the
regular random field
$F(t,x)=U(t,x)-y\,x$ and to the semimartingale $X_t=-\tU_y(t,y)$. The following identities 
will be 
useful, 
$F(t,-\tU_y(t,y))=\tU(t,y)$, $U_{xx}(t,-\tU_y(t,y))=-1/\tU_{yy}(t,y)$.\\[1mm] 
\rmi a) Observe that $F_x(t,-\tU_y(t,y))=U_x(-\tU_y(t,y))-y\equiv 0$, so that
the term in $F_x(s,X_s) dX_s$ disappears in the  It\^o-Ventzel formula; then
the diffusion random field ${\tilde \gamma}$ of $\bf \tU$ is ${\tilde
\gamma}(t,y)=\gamma(t,-\tU_y(t,y))$. Its derivative ${\tilde
\gamma}_y(t,y)=-\gamma_x(t,-\tU_y(t,y))\tU_{yy}(t,y)$ is by assumption the
diffusion characteristic of $\bf \tU_y$. Hence the covariation term is driven
by $\langle dF_x(t,x),-d\tU_y(t,y)\rangle=-\langle \gamma_x(t,x),{\tilde
\gamma}_y(t,y)\rangle dt$.\\
 b)
The  It\^o-Ventzel's formula is then reduced to,
\begin{eqnarray*}
d \tU(t,y)&-&\beta(t,-\tU_y(t,y))dt-\gamma(t,-\tU_y(t,y)).dW_t \\
&&= \frac{1}{2}U_{xx}\big(t, -\tU_y(t,y)\big)\langle d\tU_y(t,y)\rangle -\langle
\gamma_x(t,-\tU_y(t,y)).{\tilde \gamma}_y(t,y)\rangle dt\\
&& =\frac{1}{2}U_{xx}(t, -\tU_y(t,y))\|{\tilde \gamma}_y(t,y)\|^2 dt-U_{xx}(t,
-\tU_y(t,y))\|{\tilde \gamma}_y(t,y)\|^2dt\\
U_{xx}\text{- formulation}
&& =-\frac{1}{2}U_{xx}(t, -\tU_y(t,y))\|{\tilde \gamma}_y(t,y)\|^2
dt\\
\tU_{yy}\text{-formulation}&& =\frac{1}{2} \tU_{yy}(t,y)\|\gamma_x\big(t,-\tU_y(t,y)\big)\|^2dt.
\end{eqnarray*}
\rmii The dynamics of $\bf \tU_y$ is obtained
(by assumption and Theorem \ref{DRules}) by differentiating term by term  in the previous equation.
The use of coefficients $\sigma(t,y)=\gamma_x\big(t,-\tU_y(t,y)\big)$ and
$\mu(t,y)=\beta_x\big(t,-\tU_y(t,y)\big)$ of the SDE associated with ${\bf U}_x$
allows us to express $\bf \tU_y$ as the solution of a SPDE driven by the
operator t
$\wL^{\sigma,\mu}_{t,y}=\demi\partial_y(\|\sigma(t,y)\|^2 \partial_y)-\mu(t,y)
\partial_y$,
\begin{eqnarray*}
d\tU_y(t,y)&=&-\tU_{yy}(t,y)[\mu(t,y)dt+\sigma(t,y).dW_t]+\partial_y(\frac{1}{2}
\tU_{yy}(t,y)\|\sigma(t,y)\|^2)dt\\
&=& -\partial_y\tU_{y}(t,y)\sigma(t,y).dW_t+\wL_{t,y}^{\sigma,\mu}(\tU_y)dt
\end{eqnarray*}
The proof is achieved.
\end{pf} 
\begin{Remark}\label{rq: inverseSDE}{\rm 
Obviously, we are also interested in the properties of the SDE$(\tilde{\mu},\tilde{\sigma})$
 associated with the monotonic random field
 $\bf \tU_y$, 
 $\tilde{\mu}(t,z)= {\tilde \beta}_y(t,(\tU_y)^{-1}(t,z))$ and $\tilde{\sigma}(t,z)= {\tilde
\gamma}_y(t,(\tU_y)^{-1}(t,z))$.
Given that $(\tU_y)^{-1}(t,z)=-U_x(t,z)$,\\[-8mm]
\bc
 $U_{xx}(t,z) {\tilde \sigma}(t,z)=\gamma_x(t,-z)$ and $
U_{xx}(t,-z){\tilde\mu}(t,-z)= \Big(\beta_x(t,z)-\frac{1}{2}\partial_x\big(\frac{\|{\gamma}_x(t,z)\|^2}{\,U_{xx
}(t,z)}\big)\Big)$.
\ec
 It is clear that these coefficients are not globally
Lipschitz. The problem in studying directly the SDE($\tilde
\sigma, \tilde
\mu$) is the existence of a possible explosion time $\tau(x)$ as it is  shown 
in the next section \ref{subsec:regularityKunita}, Theorem \ref{DRules}. Let us first introduce  some
additional tools about regularity issues.}
\end{Remark}
\section{Regularity of It\^o's random fields and SPDEs}\label{FRSDE}
In this section, we focus on the regularity of the local characteristics of It\^os random fields in order to justify flows properties, in particular in terms of derivatives, monotony....I
We also
establish a connection between SDEs and SPDEs, that will be useful  to characterize  the market-consistent progressive utilities 
 from their dynamics.
\subsection{Regularity issues} \label{subsec:regularityKunita}
We shall discuss the regularity of an It\^o semimartingale random field
$F(t,x)=F(0,x)+\int_0^t\phi(s,x)ds+\int_0^t\psi(s,x).dW_s$
in connection with the regularity of its local characteristics $(\phi,\psi)$ and conversely. We are also concerned with the
same questions concerning SDEs solutions, where the spatial parameter is the initial condition. As in the deterministic case,
it is necessary to introduce some spatial norms very similar to Sobolev norms.\\[-9mm]

\paragraph{Definition of norms  and spaces}\label{normdef} Let $\phi$ be a continuous $\R^k$-valued progressive random field
and let $m$ be
a non-negative integer, and $\delta$ a number in $(0,1]$ . We need to control the asymptotic behavior in $0$ and $\infty$ of
$\phi$, and the regularity of its H\"{o}lder derivatives when there exist. More precisely, 
let $\phi$ be in the class ${\mathcal C}^{m, \delta}(]0,+\infty[)$, i.e.
$(m, \delta)$-times continuously differentiable in  $x$ for any $t$, a.s. \\
\rmi For any  subset  $K\subset ]0,+\infty[$, we define the family of
 random (H\"{o}lder)  $K$-semi-norms 
 \begin{equation}\label{holdernorm}
 \left\{
  \begin{array}{cll}
&\|\phi\|_{m:K}(t,\omega)=\sup_{\substack{x\in
K}}\frac{\|\phi(t,x,\omega)\|}{x}+\sum_{1\le j\le m}\sup_{\substack{x\in
K}}\|\partial_x^{j}\phi(t,x,\omega)\|\\[2mm]
& \|\phi\|_{m,\delta:K}(t,\omega)=\|\phi\|_{m:K}(t,\omega)+\displaystyle{\sup_{\substack{x,
y\in
K}}\frac{\|\partial_x^{m}\phi(t,x,\omega)-\partial_x^{m}\phi(t,y,\omega)\|}{
|x-y|^\delta}}.
\end{array}
\right.
 \end{equation}
The case $(m=0, \delta=1)$
corresponds to the local version of the Lipschitz case used in Section 1.
When $K$ is  all the domain $ ]0,+\infty[$, we simply write
$\|.\|_{m}(t,\omega)$, or $\|.\|_{m,\delta}(t,\omega)$.\\
\rmii The previous semi-norms are related to the spatial parameter. We add the temporal
dimension in assuming
 these semi-norms (or the square of the semi-norm) to be
integrable in time with respect to the Lebesgue measure on $ [0, T] $ for all $
T $. Then, as in Lebesgue's Theorem, we can differentiate, pass to the limit, commute
limit and integral for the random fields.  Calligraphic notation recalls that these semi-norms are random. 
 
 \rma $\,{\K}^{m,\delta}_{loc}$ (resp. $\overline{{\K}}^{m,\delta}_{loc}$) denotes the set of all  $\Cc^{m,\delta}$-random
fields
such that  for any compact $K\subset]0,+\infty[$, and any $T$,
$\int_0^T\|\phi\|_{m,\delta:K}(t,\omega)dt<\infty$, (resp.$ \int_0^T\|\psi\|^{
2}_{m,\delta:K}(t,\omega)dt<\infty$ ). 

\rmb When these different norms are well-defined on the whole space $]0,+\infty[$, the
derivatives (up to a certain order) are bounded in the spatial parameter, with
integrable  (resp. square integrable) in time random bound. In this case, we use the notations
${\K}^{m}_b,\overline{{\K}}^{m}_{b}$ or ${\K}^{m,\delta}_b,\overline{{\K}}^{m,\delta}_{b}.$\\[-8mm]

\paragraph{Regularity properties of random fields and SDEs}
 The following proposition is a short presentation of technical results in Kunita \cite{Kunita:01}.

\begin{Proposition}[Differential rules for random fields]\label{DRules}
Let  $\bf F$  be an It\^o semimartingale random field with local characteristics  $(\phi,\psi)$, 
 $F(t,x)=F(0,x)+\int_0^t\phi(s,x)ds+\int_0^t\psi(s,x).dW_s$\\[1mm]
\rmi If $\mathbf F$ is a ${\K}^{m,\delta}_{loc}$-semimartingale for some $m\ge0,~\delta \in (0,1]$,  its local characteristics 
$(\phi,\psi)$ are of class ${\K}^{m,\eps}_{loc}\times \overline{{\K}}^{m,\eps}_{loc}$ for any $\eps<\delta$.\\[1mm]
 \rmii Conversely, if the local characteristics $\mathbf{ (\phi,\psi)}$ are of class  ${\K}^{m,\delta}_{loc}\times
\overline{{\K}}^{m,\delta}_{loc}$,
then $\bf F$ is a ${\K}^{m,\eps}_{loc}$-semimartingale for any $\eps<\delta$.\\
\rmiii In any cases, for $m\ge 1, \delta \in (0,1]$, the derivative random field $\bfF_x$ is an It\^o random field with local
characteristics $(\phi_x,\psi_x)$.
\end{Proposition}
 The particular case of SDEs solutions is of major interest for the applications. The presentation follows \cite{MrNek01}.
\begin{Theorem}[Flows property of SDE] \rma {\bf Strong solution}  Consider a SDE$(\mu,\sigma)$, whith uniformly Lipschitz coefficients $(\mu, \sigma)\in  {\K}^{0,1}_{\bf b}\times \overline{{\K}}^{0,1}_{\bf b}$. 
There exists a unique strong solution $X$ such that 
 \\[1mm]
\centerline{$dX_t=\mu(t,X_t)dt+\sigma(t,X_t).dW_t,\quad X_0=x$.}\\[1mm]
\rmi If $ \mu \in {\K}^{m,\delta}_{\bf b}$ and $\sigma \in \overline{{\K}}^{m,\delta}_{\bf b}$ for some  $ m\ge 1,\delta\in
(0,1]$, the 
solution ${\bf X}=(X_t^x, x>0)$  is a $ \K^{m,\eps}_{loc}$ semimartingale
for any $\eps<\delta$. The
inverse ${\bf X^{-1}}$ of $\bf X$ is also of class ${\cal C}^{m}$. Then, the derivatives $\bf X_x$ and $\bf 1/X_x$ are $
{\K}^{m-1,\eps}_{loc}$-semimartingales.\\
\rmii The local characteristics of ${\bf X}$, $\lambda^X(t,x)=\mu(t,X^x_t)$ and
$\theta^X(t,x)=\sigma(t,X^x_t)$ have only local properties and belong to ${\K}^{m,\eps}_{loc}\times
\overline{{\K}}^{m,\eps}_{loc}$ for any $\eps<\delta$. \\[1mm]
\rmb {\bf Local SDEs} Assume only local property on the coefficients, 
$(\mu,\sigma)\in {\K}^{0,1}_{\bf loc}\times \overline{{\K}}^{0,1}_{\bf loc}$. \\
\rmi Then,
for any initial
condition $x$, the SDE  has a unique  maximal monotonic solution $(X^x_t)$ up to an
explosion time $\tau(x)$, and $(X^x_t)$  is a global solution if and only if the explosion time $\tau(x)$ is equal to
$\infty$
for all $x>0$ a.s..\\
\rmii If $(\mu,\sigma)\in{\K}^{m,\delta}_{loc}\times\overline{{\K}}^{m,\delta}_{loc},\>m\ge1,~0<\delta\leq 1
,$  $X_t(.)$ is of class $\Cc^{m,\eps}, \eps<\delta$ on $\{ \tau(x)>t\}$.

\end{Theorem}
 
\subsection{Solvable SPDEs via SDEs}
Since we are only concerned with non explosive solution to SDEs, we give a name to this specific class.\\ 
{\bf Class  $ \S^{m,\delta}:$} A SDE$( \mu, \sigma$) with
$(\mu,\sigma)\in{\K}^{m,\delta}_{loc}\times\overline{{\K}}^{m,\delta}_{loc}$ whose   local solution is non explosive is
said
to be of {\bf class  $ \S^{m,\delta}$}.\\
The typical example of SDE in  $\S^{m,\delta}$ is the SDE associated with the marginal conjugate utility considered as the
inverse of a solution
of SDE$(\mu,\sigma$) as in Theorem \ref{ResPrA} for  which
a  SPDE has been  associated in a very natural way in Theorem \ref{ResPrA}, under the the assumption that  the inverse flow $-{\bf \tU}_y$ of ${\bf U}_x$ is a semimartingale.  This may seem obvious, but generally the inverse of a
semimartingale is not necessarily a semimartingale. A way to define the regularity  required on the coefficients $(\mu,\sigma)$  is to formally transform the SPDE into a SDE   and to apply previous result on SDE. 
We also point that the inverse flow is less regular than the flow itself. The SPDE point of view is more efficient to calculate the stochastic transformation of the solution or of its inverse, and
 allows us to establish an exact  connection between SDEs and SPDEs.
This last point of view is well-suited  to the study of  progressive
utilities developed in this paper. \\[-6mm]
\begin{Proposition}\label{SDEInv}
Let $(X(t,x))$ be the monotonic solution of a SDE$(\mu,\sigma)$ of  class $\S^{m, \delta},m\ge
2,~\delta\in ]0,1]$,
so that as
random
field
 $(X(t,x))$ and its local characteristics $(\lambda(t,x)=\mu(t,X(t,x))$ and $(\theta(t,x)=\sigma(t,X(t,x))) $ are of class
$\K^{m,\eps}_{loc}$  and  $ {\cal
L}^{m,\eps}_{loc}\times\overline{{\K}}^{m,\eps}_{loc}$ for any $0<\eps<1$.
 We are concerned with
the SDE$( \tilde\mu,\tilde\sigma)$ 
\begin{equation}\label{EqFInv}
d\xi_t=-\frac{1}{X_x(t,\xi_t)}\Big[\big(\lambda(t,\xi_t)-\frac{1}{2}
\partial_x\big(\frac{ \|\theta\|^2}{X_x}\big)(t,\xi_t)\big)dt+
\theta(t,\xi_t).dW_t\Big],~ \xi_0=z,
\end{equation}
where 
 $\displaystyle \tilde \sigma(t,z)=-\frac{\theta(t,z)}{X_x(t,z)}$ and 
$\displaystyle \tilde \mu(t,z)=\frac{1}{X_x(t,z)}\Big(\frac{1}{2}
\partial_x\big(\frac{ \|\theta\|^2}{X_x}\big)(t,z)-\lambda(t,z)\Big). $\\[3mm]
 \rmi The SDE$(\tilde \mu,\tilde \sigma)$ is of  class $ \S^{m-2,\eps}$ $(0<\eps <\delta)$ and its unique monotonic solution
$\xi^z$ is the  inverse flow $X^{-1}$ of $X$.\\
\rmii Consequently, the inverse $X^{-1}$ of $X$ is a semimartingale and 
belongs to the class ${\K}^{m-2,\eps}_{loc}\cap {\cal C}^{m}.$
\end{Proposition}
\begin{pf}
\rmi According to Theorem \ref{ResPrA}, $\bf X$ may be considered up to a change of initial variable as a marginal progressive
utility $\bfU_x$. From Remark \ref{rq: inverseSDE}, if its inverse $\xi^X$ is "regular", then $\xi^X$ is solution of
SDE($\tilde \mu, \tilde \sigma)$ with $X_{x}(t,z) \>\tilde\sigma(t,z)=\sigma_x(t,X(t,z))$
and \\[2mm]
\centerline{$X_{x}(t,z)\> \tilde \mu(t,z)=\frac{1}{2} \partial_x\big(\frac{ \|\sigma(t,X(t,x))\|^2}{X_x(t,x)}\big)(t,z)-\mu(t,X(t,z))$}\\[2mm]
%
  The coefficients of the local  SDE($\tilde \mu, \tilde \sigma)$ are of class $ {\K}^{m-2,\eps}_{loc}\times
\overline{{\K}}^{m-1,\eps}_{loc}.$
Then, the SDE has a unique maximal solution $\xi(t,z)$ up to a life time $\tau(z)$.  It remains to show that by the
It\^o-Ventzel formula 
$X(t,\xi(t,z))=z$ on $[0,\tau(z))$.  Assume this is proven. Then the continuous (in time) process $X(t,\xi^z_t)$ is constant
a.s. on $[0,\tau(z))$.
At time $t= \tau(z)<\infty$,  $\xi(t,z)=\infty$ and $X(t,\infty)=\infty$. On the other hand, by continuity, $X(t,\xi(t,z))=z$
if  $t=
\tau(z)<\infty$.
To avoid contradiction, necessarily $\tau(z)=\infty ,\>a.s.$. So  $\bf \xi$ is the inverse flow $\bf \xi^X$ of $\bf X$. The
proof of $X(t,\xi(t,z))=z$ is
very similar to  the next proof, so we omit it here. 
\end{pf}\\
We come back now to the SPDE point of view as in Section $2$
\begin{Theorem}\label{SPDEInv}
Let us consider a SDE $(\mu,\sigma)$ of class ${\cal S}^{m, \delta}$
 with $m\geq 2, \delta \in (0,1]$, and  its adjoint operator
$\wL^{\sigma,\mu}_{t,z}=\demi\partial_z(\|\sigma(t,z)\|^2 \partial_z)-\mu(t,z) \partial_z$. Denote by  $X$  its unique
solution. \\
\rmi The inverse flow 
$X^{-1}=\xi^X$ of $X$ is a  strictly monotonic solution of  class ${\K}^{m-2,\delta}_{loc}\cap\cal C^m$ of 
 SPDE$(\widehat L^{\sigma,\mu}, -\sigma \partial_z)$, with initial condition $\xi_0(z)=z$,
\begin{eqnarray}\label{AdjointSPDE}
d\xi(t,z)= -\xi_z(t,z)\sigma(t,z).dW_t+\wL_{t,z}^{\sigma,\mu}(\xi) dt.
\end{eqnarray}
\rmii Conversely,  $(m\geq 2)$, let $\xi$ be a ${\K}^{1,\delta}_{loc}\cap{\cal C}^{2}$-regular  solution of  
SPDE$(\wL^{\sigma,\mu},-\sigma \partial_z)$ \eqref{AdjointSPDE}. Then,
$\xi(t,X(t,x))\equiv x$ and $\xi$ is  the strictly monotonic inverse flow  $X^{-1}:=\xi^X$ of $X$.
 Moreover, uniqueness  holds true for the SPDE$(\wL^{\sigma,\mu},-\sigma \partial_z)$ in the class of 
${\K}^{1,\delta}_{loc}\cap \cal C^2$-regular solutions. 
\end{Theorem}

\begin{pf}
\rmi We  start with a 
monotonic solution  $\xi$ of class $\K^{1,\delta}_{loc}\cap \cal C^2$ of the SPDE: $d\xi(t,z)=
-\xi_z(t,z)\sigma(t,z).dW_t+\wL_{t,z}^{\sigma,\mu}(\xi)
dt$. \\
\rmii From Theorem \ref{DRules}, $\xi$ is regular enough to use It\^o-Ventzel's formula with the
 solution  $X(t,x)=X^x_t$ of the SDE$(\mu,\sigma)$ to compute the dynamics of $H(t,x)=\xi(t,X(t,x))$. In the next equation,
we do not recall the parameter $x$.
\begin{eqnarray*} 
 dH_t&=&\big(-\xi_z(t,X_t)\sigma(t,X_t)-\xi_z(t,X_t)\sigma(t,X_t) \big).dW_t\\ \nonumber
 &+&\Big(\wL^{\sigma,\mu}(\xi)+  \demi \xi_{zz}\|\sigma\|^2 +\mu \>\xi_z+\partial_z(-\xi_z\sigma).\sigma\Big)(t,X_t)dt\\
 &=&\Big(\xi_{zz}\|\sigma\|^2+\demi \xi_z(\partial_z\|\sigma\|^2)-\partial_z(\xi_z)\|\sigma\|^2-\demi
\xi_z(\partial_z\|\sigma\|^2)\Big)(t,X_t)dt{\bf =0}
 \end{eqnarray*}
 The random field $H(t,x)=\xi(t,X(t,x))$ is constant in time and equal to its initial condition $x$. This finishes the proof
that $X$ is the inverse flow of $\xi$. The $\cal S^{m,\delta}$-SDE$(\mu,\sigma)$ has only one solution $X$. Then any
"regular" solution 
$\xi$ of the SPDE is the inverse of $X$ and then is unique. 
\end{pf}
\\
The next result, useful for applications, is a slight extension of the previous one.
It establishes a connection between a more general second order  SPDE and two SDEs. It is based on the observation
that
if $\xi$ is the inverse of the monotonic solution $X$ of SDE$(\mu^X,\sigma^X)$  and if $\phi\in \Cc^2$ a regular monotonic
function, the process $X(.,\phi(x))$ satisfies the same SDE$(\mu^X,\sigma^X)$, and so  its inverse 
$\phi^{-1}(\xi_.(z))$ satisfies the same SPDE than $\xi$. The extension describes the  SPDEs associated 
with the compound processes $Y(t,\xi(t,z))$, identified as the unique solution.

\begin{Theorem}\label{pp: composition formula}  Let $X$ be a solution of SDE$(\mu^X,\sigma^X)$ and $\xi$ a 
${\K}^{1,\delta}_{loc}\cap \cal C^2$-regular solution $(\delta>0)$ of the
SPDE$(\widehat L^{X},-\sigma^X \partial_z)$, where $\widehat L_{t,z}^{X}=\demi\partial_z(\|\sigma^X(t,z)\|^2
\partial_z)-\mu^X(t,z) \partial_z$.\\[1mm] 
\rmi Let $Y$ be a solution of class ${\K}^{1,\delta}_{loc}\cap \cal C^2$ of SDE$(\mu^Y,\sigma^Y)$   and 
$\phi$ any function in $\Cc^2$. Then the random field $Y(t,\phi(\xi(t,z)))=G(t,z)$ evolves as,
\begin{eqnarray}\label{eq: composition formulaA}
 dG(t,z)&=&\sigma^Y(t,G(t,z)).dW_t+\mu^Y(t,G(t,z))dt\nonumber \\
 &-& \partial_zG(t,z)\sigma^X(t,
z)\big[dW_t+\sigma^Y_y(t,G(t,z))dt]+ {\widehat L}_{t,z}^{X}(G)(t,z))dt
\end{eqnarray}
with initial condition $G(0,z)=\phi(z)$.\\
  \rmii {\bf Solvable SPDE:} Conversely, let $G$ be a  solution of class ${\K}^{1,\delta}_{loc}\cap \cal C^2$ of the
SPDE \eqref{eq: composition formulaA}; then
the process $G(t,X_t(z))$ with initial condition $\phi(z):=G(0,z)$ is solution of the SDE$(\mu^Y,\sigma^Y)$.
 If uniqueness holds true for this equation, then
$G(t,z)=Y_t(t,\phi(\xi(t,z)))$ and uniqueness also holds true for the SPDE \eqref{eq: composition formulaA}.
  \end{Theorem}
\noindent  Note the different nature of assumptions (which  may be equivalent) in the
assertions of this theorem. In $\rm (i)$, we assume that the
coefficients are regular enough such that $Y$ satisfies the It\^o-Ventzel
assumptions and such that  the inverse $ \xi $ of $ X $ is an It\^o semimartingale,
while in $\rm (ii)$ we  only suppose the existence of $X$ (without regularity), but in
return we assume the existence of a smooth solution $G$ of the SPDE 
\eqref{eq: composition formulaA}.
%
%
 
\section{Market-Consistent progressive utilities of investment and consumption}\label{sec:conso}

The notion of progressive utility is very general and  should be  specified so
as to represent more realistically the dynamic evolution of the individual
preferences of an investor in a given financial market. As in statistical learning, the utility criterium is  dynamically
adjusted to be the best given the market past information.  So, the market inputs may be viewed as a calibration universe and
gives  a  test-class of processes on which the utility is chosen to provide the best satisfaction. The market input
is described by a vector space $\GX^c$ of portfolios and consumption
incorporating feasibility and trading constraints and high liquidity. 

The existence of an
admissible strategy giving the maximal satisfaction to the investor, which will
be preserved at all times in the future, explains  {\em the martingale property} in
Definition \ref{def:conso}.  On the other hand if the strategy in $\GX^c$ fails to be
optimal then it is better  not to  make investment. The optimal strategy may be
viewed as a benchmark for the investor using the progressive utility $\bf U$.
  Once his consistent progressive utility  is defined,  an investor can then turn
to a portfolio optimization  problem  in a larger financial market or  to
calculate indifference prices. 
Before extending, in a framework with consumption, the definition of    a  consistent dynamic utility,
introduced in \cite{MrNek01,MrNek02,MrNek03} and  following  Musiela and
Zariphopoulo \cite{zar-07a,zar-07}, we first define the investment universe and the set of test processes.

\subsection{The investment universe with consumption.}\label{descmarket}
 We consider an incomplete It\^o market, defined on the filtered probability space $(\Omega,{\cal F}_{t},\mathbb{P})$ driven by the  $n$-standard Brownian motion
$W$. As usual, the market  is characterized by a short rate $(r_t)$ and a
$n$-dimensional risk premium vector 
$(\eta_t)$. 

The agent  may invest in this financial market and  is allowed to consume a part of his wealth at the progressive
rate $c_t\geq 0$.  To be short, we  give the mathematical definition of  the
class of admissible strategies $(\kappa_t, c_t)$, without specifying the risky assets. The incompleteness of the market is expressed by restrictions on the risky portfolios $\kappa_t$ constrained to live in a given progressive vector space $(\sigR_t)$.\\
To avoid technicalities, we assume throughout that all the processes
 satisfy the necessary measurability and integrability conditions such that the following formal manipulations and statements are meaningful.

\begin{Definition}[Test processes] \label{TestP} \rmi The self-financing dynamics of a  wealth process with risky portfolio $\kappa$ and consumption rate $c$  is given by
\begin{equation}\label{eq:DynamX}
dX^{\kappa,c}_t= X^{\kappa,c}_t[r_tdt +\kappa_t(dW_t+\eta_tdt)] -c_t\,dt, ~ \quad
\kappa_t \in \sigR_t.
\end{equation}
where $c$ is a positive adapted process,  $\kappa$ is a  progressive $n$-dimensional vector  measuring the
volatility vector of the wealth $X^{\kappa,c}$, such that $\int_0^T c_t + \|\kappa_t\|^2dt<\infty, a.s.$.\\
\rmii A self-financing strategy $(\kappa_t,c_t)$ is admissible if it is stopped with  the bankruptcy of the investor (when the wealth process
reaches $0$) and if  the  portfolio $\kappa$ lives in a given progressive family of vector spaces $(\sigR_t) a.s.$. \\ 
\rmiii The set of the wealth processes with admissible $(\kappa_t, c_t)$, also called test processes, is denoted by $\GX^c$. When  portfolios are starting from $x$ at time $t$, we use the notation $\GX^c_t(x).$
 \end{Definition}
\noindent The following short notations will be used extensively. 
Let $\sigR$ be a vector subspace of  $\R^n$. For any $x\in \R^n$,
$x^\sR$ is the orthogonal projection of the vector $x$
onto $\sigR$ and $x^{\perp}$ is the orthogonal projection
onto $\sigoR$. \\

The existence of  a risk premium  $\eta$ is a possible
formulation of the absence of arbitrage opportunity. Since from
\eqref{eq:DynamX}, the impact of the risk premium on the wealth dynamics only
appears through the term $\kappa_t. \eta_t $ for  $\kappa_t\in \sigR_t$, there is
a "minimal" risk premium $(\eta^\sR_t)$, the projection of $\eta_t$  on
the space $\sigR_t$ $(\kappa_t. \eta_t =\kappa_t. \eta^\sR_t )$,
to which we refer in the sequel. Moreover, the existence of $\eta^\sR$ is not
enough to insure the existence of equivalent martingale measure, since in
general we do not know if the exponential local martingale
$L^{\eta^\sR}_t=\exp(-\int_0^t \eta^\sR_s.dW_s-\frac{1}{2} \int_0^t
||\eta^\sR_s||^2\,ds)$ is a uniformly integrable martingale, density of an
equivalent martingale measure. In the following definition,  we are interested in the class of
the so-called state price density processes $Y^\nu$ (taking into account the
discount factor) who will play the same role for the progressive conjugate utility,
than the test processes $X^{\kappa,c}$ for the progressive  utility.
\begin{Definition}[State price density process] \label{SPDP} \rmi An It\^o semimartingale
$Y^\nu$ is called a state price density process if  for any admissible test process
$X^{\kappa,c}$, $(c\geq 0\kappa \in \sigR)$, $Y_.^\nu X_.^{\kappa,c} + \int_0^. Y^\nu_s c_s ds $ is a local martingale. It
follows that $Y^\nu$ satisfies,
\begin{equation}\label{Ynu}
dY^\nu_t=Y^\nu_t[-r_tdt+ (\nu_t-\eta^\sR_t).dW_t],\quad
\nu_t \in \sigoR_t, \quad Y^\nu_0=y
\end{equation}
\noindent
\rmii Denote $\GY$ the convex family of all state density processes $Y^\nu$ where $~\nu\in
\sigoR$ and observe that
 $Y^\nu$ is the product of $Y^0$ $(\nu=0)$ by the density martingale
$L^{\nu}_t=\exp\big(\int_0^t \nu_s.dW_s-\frac{1}{2}\int_0^t ||\nu_s||^2ds\big)$.
\end{Definition}
 Interesting discussion on the links between these assumptions and the market numeraire $N_t=(Y^0)^{-1}_t$, also called GOP (growth optimal portfolio) can be found in the book by D.Heath \& E.Platen \cite{Platen} and in D.Filipovic\& E.Platen \cite{FilPlat2009}.
Nevertheless, the use of change of numeraire in our framework is reported in order to limit the size of the paper.

\subsection{$\GX^c$-consistent Utility and Portfolio optimization with consumption}\label{subsec:4.2}

As we are interested in optimizing both the terminal wealth and the consumption rate, we introduce two progressive utilities: the first
one, 
$\bf U$, for the terminal wealth and the second one, $\bf V$, for the consumption rate. From a dynamic point of view,
 $\bf U$ and $\bf V$ will play different roles, only $\bf U$ will need to be an It\^o progressive utility.  To express that the adaptative criteria $(U,V)$ are  well-adapted to the investment universe,
 we introduce the following conditions: 
\begin{Definition}\label{def:conso}
A $\GX^c$-consistent progressive utility  system of investment and consumption
is a pair of progressive utilities $\bf U$ and $\bf V$ on $\Omega\times[0,+\infty)\times\R^+$
with the following additional properties:\\
\rmi {\sc Consistency with the test-class:}  For
any admissible wealth process  $X^{\kappa,c}\in \GX^c$, and any couple of dates $t<T$,
\begin{equation*}
 \mathbb{E}\big(U(T,X^{\kappa,c}_T)+\int_t^T V(s,c_s)ds /{\cal F}_t\big)\leq
U(t,X^{\kappa,c}_t), \>~ 
a.s. 
\end{equation*}
In other words, the value process $\b({\cal G}^{\kappa,c}_t=U(t,X^{\kappa,c}_t)+\int_0^tV(s,c_s)ds\b)$ is a positive supermartingale.\\
\rmii {\sc Existence of optimal strategy:} 
For any initial wealth $x>0$, there exists an optimal strategy $(\kappa^*,c^*)$ such that the associated non negative
wealth process $X^{*}=X^{\kappa^*,c^*}\in \GX^c$ issued from $x$ satisfies  $\b({\cal G}^{*}_t=U(t,X^*_t)+\int_0^tV(s,c^*_s)ds\b)$ is a local
martingale.
 \end{Definition}

\paragraph{$\GX^c$-consistent It\^o progressive utilities and HJB constraint}
Theorem \ref{UxSDE} characterizes It\^o
progressive utilities in terms of their local characteristics 
$(\beta, \gamma)$ as well as in terms of the parameters $(\mu, \sigma)$ of the intrinsic SDE \eqref{eq:ItoSDE} satisfied by
$U_x$. In this section, we are concerned with the constraints induced on the  characteristics $(\beta, \gamma)$ of 
$\bf U$
 by the  $\GX^c$ consistency property.

\noindent The supermartingale/martingale property  of processes ${\cal G}^{\kappa,c}_t$ implies 
negative drift for these It\^o processes 
for all $\kappa \in  \sigR, c\geq 0$, and  $0$ drift for some $(\kappa^*,c^*)$. This property yields to Hamilton-Jacobi-Bellman  type constraints on the drift $\beta(t,x)$ of $\bf U$.\\
We proceed by verification as usual by introducing a non standard Hamilton-Jacobi-Bellman Stochastic PDE.
 Observe that the consumption optimization contributes only through the Fenchel-Legendre
random field $\tilde{\bf V}$ of the dynamic utility $\bf V$.

\begin{Theorem}[Utility-SPDE]\label{EDPGacons}
Let $(\bf U,\bf V)$  be a utility system where $\bf U$ is  regular enough to apply It\^o's Ventzel formula. Define a  monotonic random field
$\bar \zeta(t,x)$ by $\bar \zeta(t,x)=(V_c)^{-1}(t,U_x(t,x))=-\tilde V_y(t,U_x(t,x))$, and a  "policy" random field $x \bar \kappa_t(x)$ by 

\begin{equation}\label{portconso}
x \bar \kappa_t(x)=-\frac{1}{\,U_{xx}(t,x)}\big(U_x(t,x)\eta^\sR_t+\gamma_x^\sR(t,x)\big).
\end{equation}
 Then the $\GX^c$- consistency property of $(\bf U,\bf V)$ is implied by the following two assertions: \\
\rma The drift $\beta$ satisfies the following HJB-constraint
\begin{eqnarray}\label{BETAcons}
\beta(t,x)=-U_x(t,x)xr_t+\frac{1}{2}U_{xx}(t,x)\|x \bar \kappa_t(x)\|^2-\tilde V(t,U_x(t,x)).
\end{eqnarray}
\rmb The 
 SDE$(\bar \mu^{c}, \bar \sigma)$, with $\bar \mu^{c}_t(x):=r_t x+x \bar \kappa_t(x).\eta_t^\sR-\bar \zeta(t,x)$       and   $\bar \sigma(t,x):= x \bar \kappa_t(x)$
admits  a non negative solution $\bar X \geq 0$. 
Furthermore,  $(x \bar\kappa(x),\bar c_t=\bar \zeta(t,\bar X_t))$ is the  optimal strategy of investment and consumption (in general denoted with a $(*)$) with monotonic optimal wealth $X^*=\bar X$.
\end{Theorem}
\begin{proof}
\rmi  By It\^o-Ventzel's formula (Theorem \ref{IVF}), for any test process
$X^{\kappa,c}$,
\begin{eqnarray*}
dU(t,X^{\kappa,c}_t)+V(t,c_t)dt &=&\Big(U_x(t,X^{\kappa,c}_t)X^{\kappa,c}_t\>
\kappa_t+\gamma(t,X^{\kappa,c}_t)\Big).dW_t\\
&+&\Big(\beta(t,X^{\kappa,c}_t)+U_x(t,X^{\kappa,c}_t )r_t
X^{\kappa,c}_t+\frac{1}{2}U_{xx}(t,X^{\kappa,c}_t)\mathcal{Q}(t,X^{\kappa,c}_t,
\kappa_t)\Big)dt\\
&+&\big(V(t,c_t)-U_x(t,X^{\kappa,c}_t)c_t\big)dt\\
\mbox{where}\>\mathcal{Q}(t,x,\kappa)&=&\|x \kappa\|^2+2 x
\kappa.\big(\frac{U_x(t,x)\eta^\sR_t +
\gamma_x(t,x)}{\,U_{xx}(t,x)}\big).
\end{eqnarray*}
Since $\kappa \in \sigR$, $\mathcal{Q}(t,x,\kappa)$ is only depending on 
$\gamma_x^\sR(t,x)$, the orthogonal projection of $\gamma_x(t,x)$ on
$\sigR_t$. 
The minimum $\mathcal{Q}^*(t,x)= \inf_{\kappa \in \sigR}\mathcal{Q}(t,x,\kappa)$
of the quadratic form $\mathcal{Q}(t,x,\kappa)$ is achieved at the  minimizing policy
$\bar \kappa $  given by
\begin{equation}\label{kappaQ*}
\left\{
\begin{array}{cll}
x \bar\kappa_t(x) &=-\frac{1}{\,U_{xx}(t,x)}\big(U_x(t,x)\eta^\sR_t
+\gamma_x^\sR(t,x)\big)\\[1mm]
\mathcal{Q}^*(t,x)&=-\frac{1}{U_{xx}(t,x)^2}\|U_x(t,
x)\eta^\sR_t+\gamma_x^\sR(t,x))\|^2=\,-\|x \bar \kappa_t(x) \|^2.
\end{array}
\right.
\end{equation}

\rmii  By the Fenchel convexity inequality, the term in the third line is bounded by above by   $V(t,c_t)-U_x(t,X^{\kappa,c}_t)c_t\leq \tilde V(t,-U_x(t,X^{\kappa,c}_t))$.\\
For the second line, since  $\mathcal{Q}(t,x,\kappa_t)\geq  \mathcal{Q}^*(t,x)= - \|x \bar \kappa_t(x) \|^2$ and $U_{xx}\leq 0$,
 the term may be bounded by above by
$(\beta(t,X^{\kappa,c}_t)+U_x(t,X^{\kappa,c}_t )r_t
X^{\kappa,c}_t-\frac{1}{2}U_{xx}(t,X^{\kappa,c}_t) \|X^{\kappa,c}_t \bar \kappa_t(X^{\kappa,c}_t ) \|^2  $. \\

Then, if $\beta$ satisfies the HJB constraint (\ref{BETAcons}),  the drift term is nonpositive  for any $\kappa
\in \sigR$ and $c\ge0$,  and the process $\big(U(t,X^{\kappa,c}_t)+\int_0^tV(s,c_s)ds\big)$ is a
supermartingale. \\
\rmiii  Assume now that the wealth SDE associated
with $(\bar \kappa, \bar \zeta)$ admits a positive solution $\bar X$. Then, the non positive drift in the previous equation is equal to $0$, so that $\big(U(t,\bar X_t)+\int_0^t V(s, \bar \zeta(s,\bar X_s ))ds\big)$ is a local
martingale. This equality proves the existence of an optimal strategy $x\bar \kappa$ and $\bar\zeta$, and that $\bar X$ is an optimal process. So, we have no  reason to distinguish between processes with $^-$ and processes with $^*$.
\end{proof}

\paragraph{Conjugate of consistent progressive utility with consumption. }
\noindent  Let $(\bfU,\bfV)$ be a pair of  stochastic $\GX^c$-consistent utilities with optimal strategy $(\kappa^*,c^*)$
leading to the non negative
wealth process $X^{*}=X^{\kappa^*,c^*}$. 
Convex analysis  showed  the interest to study the convex conjugate utilities $\tilde{\bf U}$  and  $\tilde{\bf V}$.  
Indeed, under mild regularity assumption, we have the following results (Karatzas-Shreve \cite{KaratzasShreve:01}, Rogers \cite{Rogers}).
\begin{itemize}
\item[(i)] For any admissible state price density process $Y^\nu\in \GY$ with $\nu\in \sigoR$,
$\Big(\tU(t,Y^\nu_t)+\int_0^t\tV(s,Y^\nu_s)ds\Big)$ is a  submartingale, and there exists a unique optimal process $Y^*:=Y^{\nu^*}$  with $\nu^* \in \sigoR$ such that  $\Big(\tU(t,Y^*_t)+\int_0^t\tV(s,Y^*_s)ds\Big)$
is a local martingale. To summarize $\U(s,Y_s)=\esssup_{Y\in \GY}\E\Big(\big(\tU(t,Y^\nu_t)+\int_0^t\tV(\alpha,Y^\nu_\alpha)d\alpha\big)/\F_s
\Big)$, \quad a.s.
\item[(ii)] {\it Optimal Processes characterization} Under regularity assumption, first order conditions imply some links between optimal processes, including their initial conditions,
\begin{eqnarray}\label{first order condition}
 Y^*_t(y)=&U_{x}(t,X^*_t(x))=  V_{c}(t,c^*_t(c_0)), \quad &y=u_x(x)=v_c(c_0).
\end{eqnarray}
\end{itemize}
The characteristics of the consistent conjugate progressive utility $\bf \tU$ 
with consumption can be also  computed directly from Theorem \ref{ResPrA}, using a PDE approach.
Given  that the drift 
$\beta$ is associated with an optimization program, it is easy to show that
$\tilde{\beta}$ is also constrained by  a HJB type relation in the new
variables,  and the
convex conjugate utility system $(\bf \tU,\bf \tV)$ is  consistent (in a sense to be precised) with a family of state price
density processes (Definition \ref{SPDP}). 
\begin{Theorem}\label{thEDPDuale}
Let $(\bfU,\bfV)$  a consistent progressive utility system with consumption, such that  $U$ is 
$\K^{3,\delta}_{loc}$-regular $(\delta>0)$  with  local characteristics $(\beta,\gamma)$ satisfying  Assumptions of
Theorem \ref{EDPGacons}. Then\\
\rmi The progressive convex conjugate utility $\bf \tU$ and its marginal conjugate utility ${\bf \tU_y}$ are It\^o
random fields with local characteristics $(\tilde\beta, \tilde\gamma)$ and $(\tilde\beta_y, \tilde\gamma_y)$ respectively.\\
 \rmii The local characteristics of the convex conjugate $\bf \tU$  are given by:
 \begin{eqnarray}\label{EDPSDuale'}
\left\{
\begin{aligned}
&~\tilde{\gamma}(t,y):=\gamma(t,-\tU_y(t,y)), \quad
\tilde{\gamma}_y(t,y):=-\gamma_x(t,-\tU_y(t,y)). \tU_{yy}(y)\\
&\tilde{\beta}(t,y)=y\tU_{y}(t,y) r_t+\frac{1}{2}\tU_{yy}(t,y)\|\tilde \sigma^*(t,y)\|^2-
 \sigma_t^*(-\tU_{y}(t,y)).y\eta^\sR_t-\tV(t,y)
\end{aligned}
\right .
\end{eqnarray} 
\rmiii For any admissible state price density process $Y^\nu\in \GY$ with $\nu\in \sigoR$,
$\Big(\tU(t,Y^\nu_t)+\int_0^t\tV(s,Y^\nu_s)ds\Big)$ is a
 submartingale, and a local martingale for any solution $Y^*$ (if there exists) of the equation $dY^*_t=Y^*_t[-r_t
dt+(\nu^*(t,Y^*_t)-\eta^\sR_t).dW_t]=\tmu(t,Y^*_t)dt-\tsigma(t,Y^*_t).dW_t$,
with $\tilde \sigma^*(t,y)=y(\nu_t^*(y)-\eta^\sR_t)$ and  $ \tmu(t,y)=-r_t \,y.$
 \end{Theorem}
 \begin{proof} A similar proof in the framework without consumption can be found in \cite{MrNek01}. \end{proof}

\noindent From now on, either the notation $\sigma^*_t(y) $ or $\sigma^*(t,y)$ will be used. To fiw the idea, we now give the example of consistent  Power Utilities, for which we prove the existence of optimal processes without any additional regularity conditions. 

\subsection{ Consistent  Power Utilities}\label{sec:expuissance}
Power utilities with constant risk aversion are widely used in economics, in particular for the Ramsey rule established in the next Section. It is also a useful example in the framework of  forward utilities for its simplicity and its easy interpretation of the coefficient. To characterise such utilities,
we start with a problem without consumption.\\[-9mm]
\paragraph*{Consistent Progressive  Power Utility without consumption}  see \cite{MrNek02} for more details.\\
\rmi 
Let us consider a consistent power utility
$U^{(\alpha)}(t,x)=Z^{(\alpha)}_t \frac{x^{1-\alpha}}{1-\alpha}\>$  where $\alpha\in (0,1)$ is the risk aversion coefficient 
 and $Z^{(\alpha)}$ a semimartingale allowing
 to satisfy the consistency property. Then, the conjugate function $\tU^{(\alpha)}(t,y) $ satisfied $\tU^{(\alpha)}(t,y) =\tilde{Z}^{(\alpha)}_t\frac{y^{1-\frac{1}{\alpha}}}{\frac{1}{\alpha}-1}$. Since  the risk aversion coefficient is given, we do not recall it if it not necessary.\\
\rmii
Thanks to  the consistency property, there exists an optimal wealth process $X^*(x)$ such that 
$U^{(\alpha)}(t,X^*_t(x))=\frac{1}{1-\alpha}\,Z_t\> \b(X^*_t(x)\b)^{1-\alpha}$ is a martingale, and such that
$U^{(\alpha)}_x(t,X^{*}_t(x))=Y^{*}_t(x^{-\alpha})$ is a state price density process with initial condition $x^{-\alpha}$. \\
\rmiii In particular, using the intuitive factorization $Z_t=Z^{\sigR}_t.Z^{\perp}_t$ where $Z^{\perp}_t$ is an
exponential martingale $\mathcal{E}_t(\delta^{\perp}.W)$  with $\delta^{\perp} \in \sigoR$,

 we see that
$Z^{\sigR}_t (X^*_t(x))^{-\alpha}=x^{-\alpha}Y_t^0$, where $Y_t^0$ is the minimal state price density.\\ {\em The optimal wealth
$X^{*}_t(x)$ and the optimal dual process $Y^{*}_t(y)=yY^{*}_t$ are linear  with respect to their initial condition.} So,

\begin{equation*}
(I)\left\{
\begin{aligned}
Z_t&=&Z^{\perp}_t\,Y_t^0\,(X^{*}_t)^{\alpha},\hspace{15mm} &\tilde{Z}_t=X^{*}_t\big(Y^{*}_t\big)^{\frac{1}{\alpha}}\\
X^{*}_t(x)&=&x X^{*}_t,\hspace{30mm}&Y^{*}_t(y)=yY^{*}_t=yZ^{\perp}_t\,Y_t^0\hspace{15mm}\\
U^{(\alpha)}(t,x)&=&\frac{Y^*_t X^*_t}{1-\alpha}\b(\frac{x}{X^*_t}\b)^{1-\alpha},\hspace{10mm} &
\tU^{(\alpha)}(t,y) =\frac{Y^*_tX^*_t}{\frac{1}{\alpha}-1}\big(\frac{y}{Y^*_t}\big)^{1-\frac{1}{\alpha}}
\end{aligned}
\right.
\end{equation*}
Since the characteristics of the power utility $U^{(\alpha)}(t,x)=Z_t \frac{x^{1-\alpha}}{1-\alpha}\>=Z_t \>u^{(\alpha)}(x)$ are only dependent of the characteristics $(\beta^Z, \gamma^Z)$ of $Z$, \\[1mm]
\centerline{$\gamma^{(\alpha)}(t,x)=\gamma^Z_t \>u^{(\alpha)}(x), \quad  \beta^{(\alpha)}(t,x)=\beta^Z_t \>u^{(\alpha)}(x)$, \quad with \quad $dZ_t =\beta^Z_t dt+\gamma^Z_t.dW_t.$ \quad \quad }\\[1mm]
Equations \eqref{portconso} and  \eqref{BETAcons} of Theorem \ref{EDPGacons} are easily verified (with $V=\tilde V=0$), from the formula $Z_t=Y^{*}_t\,
(X^{*}_t)^{\alpha}$ whose differential characteristics are $\gamma^Z_t=Z_t (\alpha \kappa^*_t+(\nu^*_t-\eta_t^\sigR))$ and $\beta^Z_t=(1-\alpha)Z_t(-r_t+ \demi \alpha \|\kappa_t^*\|^2).$
%
\\[-8mm]
\paragraph*{Consistent Progressive Power  Utility with consumption}  
When the problem consists in optimizing also  a consumption process,  we have to precise what stochastic utility for the consumption we must choose to satisfy the consistency of the utility system $(U^{(\alpha)}, V)$ when $U^{(\alpha)}_t(x)=\widehat Z_t u^{(\alpha)}(x)$ is a power progressive utility, and $\widehat Z$ a semimartingale with local characteristics ($\hat \gamma, \hat \beta)$. A useful tool is the system of equations (\ref{eq:HJBConst}), since the equation characterizing the  process $\gamma_x(t,x)=\widehat \gamma_t u_x^{(\alpha)}(x)$ does not depend explicitly on $V$. To make the distinction between the two problems, we introduce the symbol 
$\hat.$ in the quantities relative to the problem with consumption. \\
\rmi Equation \eqref{portconso}, after dividing the both sides by $u_x^{(\alpha)}(x)$, yields to 
\bc $\widehat \gamma_t =\widehat Z_t \b(\alpha x^{-1}\widehat \sigma^*_t(x)+( \widehat \nu^*_t(U_x(t,x)) -\eta_t^\sigR).$
\ec
Since $\widehat \gamma_t $ does not depends on $x$, this equality implies as above that $x^{-1}\widehat \sigma^*_t(x)=\widehat \kappa^*_t(x)$ and $\widehat \nu^*_t(U_x(t,x))$ does not depend on $x$, so as in the situation without consumption $\widehat \gamma_t=\widehat Z_t (\alpha \widehat \kappa^{*}_t+(\widehat \nu^*_t-\eta_t^\sigR))$.\\
\rmii The drift equation  (\ref{BETAcons}) becomes 
 \bc $\widehat \beta_t=\widehat Z_t(-(1-\alpha)\,r_t+ \demi \alpha(1-\alpha)\|\widehat \kappa^{*}_t\|^2)-\widetilde V(t,\widehat Z_t u_x^{(\alpha)}(x))/u^{(\alpha)}(x)$.
 \ec
Thus, by the same argument as before, the progressive conjuguate utility $\widetilde V(t,y)$ must be chosen in such a way that 
$\widetilde V(t,\widehat Z_t u_x^{(\alpha)}(x))=\alpha \>\widehat \psi_t \> \widehat Z_t \>u^{(\alpha)}(x)$, where $\widehat \psi_t$
 is a positive adapted process with good integrability property.  As a consequence, $\widetilde
V(t,y)=\widehat \psi_t (\widehat Z_t )^{1/\alpha} \>\tilde u^{(\alpha)}(y).$ So, $\widetilde V$ is the
Fenchel transform of  a power utility $V(t,x)=(\widehat \psi_t )^\alpha \widehat Z_t \>u^{(\alpha)}(x)= (\widehat \psi_t
)^\alpha\,U^{(\alpha)}(t,x)$.\\
\rmiii Then, the process $\widehat Z$  is a solution of the stochastic differential equation 
\bc
$ \hspace{5mm}	d\widehat Z_t= \widehat Z_t\b[(\alpha \widehat \kappa^{*}_t+(\widehat \nu^*_t-\eta_t^\sigR).dW_t+\b(-(1-\alpha)\,r_t+ \demi \alpha(1-\alpha)\|\widehat \kappa^*_t\|^2-\alpha \> \widehat \psi_t \b)dt\b] $
\ec
 where the optimal strategies are the same as in the case without consumption ($\widehat \kappa^{*}\equiv \kappa^{*},\>\widehat \nu^*\equiv \nu^*)$, 
$\widehat Z_t= Z_t e^{-\alpha\int_0^t \widehat \psi_s  ds}$. The process $\widehat \psi_t$  plays in this formula the  role of an additional spread to the interest rate $r_t$  for the wealth but not for the sate price density. This interpretation is justified by the closed form of the optimal consumption $c^*_t(x)=-\widetilde V_y(t,\widehat Z_t u_x^{(\alpha)}(x))$ given after some tedious calculation by $c^*_t(z)=z\,\widehat \psi_t$. 
\begin{Corollary} \label{consforwardpower}A consumption consistent progressive power utility system is  necessarily a pair of power utilities with the same risk aversion coefficient $\alpha$ such that \\[1mm]
\centerline{$U^{(\alpha)}(t,x)=\widehat Z_t \frac{x^{1-\alpha}}{1-\alpha}=\widehat Z_t u^{(\alpha)}(x)\quad$ and $\>V^{(\alpha)}(t,x)=(\hat \psi_t)^\alpha U^{(\alpha)}(t,x)$.  }\\[1mm]

\rmi The  optimal processes are linear with respect of their initial condition, i.e. 
\\[1mm]
\centerline{$\widehat X^*_t(x)=x \widehat X^*_t, \> Y^*_t(y)=y Y^*_t, \>\text{and}\quad c^*_t(z)=z\>\widehat \psi_t.$}\\
\rmii The coefficient  $\widehat Z_t$ is determined by the optimal processes via $\widehat Z_t=Y^*_t (\widehat X^*_t)^\alpha$, while the coefficient $\widehat \psi_t$
 is only assumed to be positive.\\
 \rmiii The optimal processes (with initial condition $1$) are driven by the system  $c^*_t=\widehat \psi_t$ and 
\begin{equation}
d\widehat X^*_t=\widehat X^*_t\b((r_t-\widehat \psi_t)dt+ \kappa_t^*.(dW_t+\eta^\sigR_t)\b), \> dY^*_t=Y^*_t\b(-r_t dt +(\nu^*_t-\eta^\sigR_t)dW_s\b)
\end{equation}
\end{Corollary}
\noindent In the general case of consistent progressive utilities, additional regularity conditions are needed, but it is still possible to give a closed form of the forward utility in terms of initial condition and optimal processes.

\subsection{Regularity issues for existence of consistent progressive utility and closed form characterization via optimal processes.}
In Subsection \ref{subsec:4.2}, we have assumed the consistent progressive utility $\bf U$ sufficiently regular to apply It\^o's Ventzel formula in view of establishing HJB constraint; then we have shown the links between local characteristics and coefficients of the SDE associated with an optimal portfolio, without proving the existence. The same kind of assumptions are made on the conjugate
$\tilde{\bf U}$, implying the dual HJB constraint in the same way  than for the primal problem. But it is well-known that in all generality these assumptions are not satisfied.\\
Assuming the existence of regular progressive utility satisfying HJB constraint, we show that 
 $(\bf U, \bf V)$ is a $\GX^c$-consistent stochastic utility system, associated with a regular optimal dual SDE$( \tilde \mu^*, \tilde \sigma^*)$ whose coefficients are based only on the diffusion characteristics $\gamma$ of $\bfU$, and do not depend on the  utility of consumption process $V$. 
 The existence of this strong dual solution is very important in view to apply Theorem \ref{pp: composition formula} not directly to $\bf U$ but to ${\bf U}_x$ whose the diffusion characteristic $\gamma_x$ has the same form than the diffusion characteristic of the random field $G$, where $\sigma^Y$ is replaced by  $\tilde \sigma^*_t$, and $\sigma^X$  by $\sigma^*_t(x)=\kappa^*_t(x)$. 
 In addition to consistency, under this HJB constraint, we show that such  utility system can be represented in a
closed  form. 
\\[1mm]
 To be closed to the notation of Theorem \ref{pp: composition formula}, we recall all the coefficients of optimal SDEs  associated
with the primal and dual problems,
\begin{equation}\label{eq:optcoeff}
\left\{
\begin{array}{llll}
\tilde \sigma^*_t(y)&:=y(\nu_t^*(y)-\eta^\sR_t), &\tmu_t(y):=-r_t \,y,  \quad \\[1mm]
\sigma^*_t(x)&:=x\kappa^*_t(x)&\mu^{*,c}_t(x):=r_t x+x\kappa^*_t(x).\eta_t^\sR-\zeta^*(t,x),\hspace{22mm}  \\[1mm]
\widehat L^{(\mu^*, \sigma^*,c)}_{t,x}&:=\frac{1}{2}\partial_x(\|\sigma^*_t(x)\|^2 \partial_x)-\mu^{*,c}_t(x)\partial_x,&
L^{(\tmu, \sigma^*)}_{t,x}=\demi \|\sigma^*_t(x)\|^2\partial_{xx}+\tmu_t(x)\partial_x
\end{array}
\right.
\end{equation}
%
%

\begin{Proposition}\label{ppU_xconso}
 Let ${\bf U}$ be a $\K^{2,\delta}_{loc}$-regular $(\delta>0)$ progressive utility ${\bf U}$,
whose  local characteristics $(\beta,\gamma)$ satisfy  the HJB constraints,
\begin{equation}\label{eq:HJBConst}
\left\{
\begin{array}{llll}
\gamma_x(t,x)&:=-U_{xx}(t,x)x\kappa^*_t(x)+\tilde \sigma^*_t(U_x(t,x)), \quad \gamma^\perp_x(t,-\tU_y(t,y))=y\nu^*_t(y)\\
\beta(t,x)&:=-U_x(t,x)x\,r_t-\tV(t,U_x(t,x))+\frac{1}{2}U_{xx}(t,x)\|x\kappa^*_t(x)\|^2
\end{array}
\right.
\end{equation}
\rmi The marginal utility ${\bf U}_x$ is a decreasing solution of the SPDE\eqref{eq:
composition formulaA}  with coefficients $(\mu^{*,c}, \sigma^*)$ and $(\tilde \mu^*,\tilde \sigma^*)$\\[-7mm]
\begin{eqnarray}\label{eq: composition formulaU}
 dU_x(t,x)&=&\tilde \sigma^*_t(U_x(t,x)).dW_t+\tilde \mu^*_t(U_x(t,x))dt \nonumber\\
 &-&\partial_xU_x(t,x)\sigma^*_t(x).\big(dW_t+\tilde{\sigma}^*_y(t,U_x(t,x))dt)+\wL^{*,c}_{t,x}(U)dt
\end{eqnarray}
\rmii  Assume that the SDE$(\mu^{*,c},\sigma^*_t)$ and SDE$(-r_t y, \tilde \sigma^*(t,y))$ admit a monotonic solution $(X^*_t(x), Y^*_t(y)).$
Then, the marginal forward utility at time $t$ is the non linear transportation of the marginal utility at time $0$ through the optimal dual processes,
\begin{equation}\label{eq:margutrepres}
U_x(t,x)=Y^*_t(u_x((X^*_t)^{-1}(x)))
\end{equation}

\end{Proposition}

\begin{proof} 
 First, as ${\bf U}$ is assumed to be $ \K^{2,\delta}_{loc}\cap \Cc^3$-regular, ${\bf U}_x$ is of class
$\K^{1,\delta}_{loc}$ and
its
local characteristics $(\beta_x,\gamma_x)$ are of class  $\Cc^1$ in $x$; then, the vectors $
\sigma^*_t(x)=-(\gamma_x^\sR(t,x)+\eta^\sR _t U_x(t,x))/U_{xx}(t,x)$ and $\tilde \sigma^*_t(y)=
\gamma_x^\perp(t,-\tU_y(t,y))-y\eta^\sR_t$ are also of class $\Cc^1$, necessary condition to define $\widehat L^{*,c}$. \\
 By derivation of  the local characteristics of the regular progressive utility ${\bf U}$, we see that\\[1mm] 
\centerline{$\beta_x(t,x)=- \partial_x(U_x(t,x)xr_t)-U_{xx}(t,x)\tV_y(t,U_x(t,x))+\partial_x\big(\demi
U_{xx}(t,x)\|\sigma^*_t(x)\|^2\big).\qquad$}\\[2mm]
Observing that $\tV_y(t,U_x(t,x))=-(V_c)^{-1}(t,U_x(t,x))=-\zeta^*(t,x)$, it follows that \\[-4mm]
$$
\begin{array}{clll}
- \partial_x(U_x(t,x)xr_t)-U_{xx}(t,x)\tV_y(t,U_x(t,x))&=- \partial_x(U_x(t,x)xr_t+U_{xx}(t,x)\zeta^*(t,x)\\
&=-\partial_xU_x(t,x)\mu^{*,c}(t,x)-r_tU_x(t,x).
 \end{array}
 $$
It remains to make some slight transformations on the drift characteristic: \\[-4mm]
 \begin{equation*}\label{pf:ppU_x}
\begin{array}{llll}
\beta_x(t,x)&=- \partial_xU_x(t,x)\mu^{*,c}(t,x)-r_tU_x(t,x)        +\partial_x\big(\demi U_{xx}(t,x)\|\sigma^*_t(x)\|^2\big)\\[1mm]
&=\widehat L^{*,c}_{t,x}(U)-r_tU_x(t,x)+\partial_xU_x(t,x) \sigma^*_t(x).\eta_t^\sR\\[1mm]
&=\widehat L^{*,c}_{t,x}(U)+\tilde \mu^*_t(U_x(t,x)) +\partial_xU_x(t,x) \sigma^*_t(x).\eta_t^\sR
\end{array}
\end{equation*}
Let us give another interpretation of $ \sigma^*_t(x).\eta_t^\sR$. Since  $\tilde
\sigma^*(t,y)+\eta^\sR_t y$ belongs to the vector space $\sigoR_t$ 
 the spatial derivative $\tilde \sigma^*_y(t,y)+\eta^\sR_t$
 is also in $\sigoR_t$, yielding to the relation on the scalar products $-
\sigma^*_t(x).\eta_t^\sR=\sigma^*_t(x).\tilde\sigma^*_y(t,y)$. Then, Identity \eqref{eq: composition formulaU} holds true.\\
\rmii If we know the existence of monotonic solution of  SDE$(\mu^{*,c},\sigma^*_t)$ and SDE$(-r_t y, \tilde \sigma^*(t,y)),$ from the form of the SPDE associated with $U_x$ and the assertion $(ii)$ of Theorem \ref{pp: composition formula}, we easily obtained the representation   $U_x(t,x)=Y^*_t(u_x((X^*_t)^{-1}(x)))$.
\end{proof}

 The next theorem gives sufficient condition for the existence of (monotonic) optimal solutions for the optimisation problem.
\begin{Theorem} Let ${\bf U}$ be a $\K^{2,\delta}_{loc}\cap\Cc^3$-regular $(\delta>0)$ progressive utility ${\bf U}$,
whose  local characteristics $(\beta,\gamma)$ satisfy HJB constraints \ref{eq:HJBConst}
.\\
{\sc Main result} Suppose the existence of two adapted bounds $(K^1,K^2)\in \L^2(dt)$ such that the regular random field $\gamma^\perp_x$ satisfy  
\begin{equation}\label{eq:gammaderivatives}
\|\gamma^\perp_x(t,x)\|\leq
K^1_t\,|U_x(t,x)|,  \> \|\gamma^\perp_{xx}(t,x)\|\leq K^2_t \,|U_{xx}(t,x)|, \>a.s., 
\end{equation}
\rmi As $ y \nu_t^*(y)=\gamma^\perp_x(t,U_x^{-1}(t,y))$, and $\tilde \sigma^*_t(y):=y(\nu_t^*(y)-\eta^\sR_t)$, the SDE$(-r_t y, \tilde \sigma^*(t,y))$ is uniformly Lipschitz and its
 unique strong solution $Y^*_t(y)$ is increasing, with range $[0,\infty).$\\
\rmii  Moreover, assume the existence of an adapted bound $K^3$ such that process $V_c(t,K^3 x) \ge U_x(t,x) \>a.s.$ for
any $x$. Using the notations $\sigma^*_t(x):=x\kappa^*_t(x)$ and $\mu^{*,c}_t(x):=r_t
x+x\kappa^*_t(x).\eta_t^\sR-\zeta^*(t,x)$, 

a) The SDE$(\mu^{*,c}, \sigma^*)$ is locally Lipschitz and admits a
maximal positive monotonic solution $X^*$ such that  $U_x(.,X^*_.(x))$ is distinguishable from the solution $Y^*_.(u_x(x))$.

b) The optimal consumption along the optimal wealth process   is\\[-7mm]
\bc
$c^*_t(x)=\zeta^*(t,X^*_t(x))=-\tV_y(t,U_x(t,X^*_t(x)))=-\tV_y(t,Y^*_t(u_x(x))$. 
\ec
 {\sc Reverse solution} Denote by $\mu^*(t,x)=r_t x+x\kappa^*_t(x).\eta_t^\sR$ the drift of some portfolio without consumption, by $\bar \zeta(t,x)$ some increasing adapted positive random field, and by ${\bar \mu}(t,x)=\mu^*(t,x)-\zeta(t,x)$. Assume the existence  $({\bar X},Y^*)$ of two monotonic solutions of SDE(${\bar \mu},\sigma^*)$ and SDE($\tmu,\tsigma$) with range $(0,\infty)$.\\
\rma For any determinisitc utility function $(u,v)$ such that $v_c(\zeta(0,x))=u_x(x)$, \\[1mm]
\centerline{
$U_x(t,x)=Y^*_t(u_x({\bar X}_t^{-1}(x)), \quad V_c(t,c)=U_x(t,\zeta^{*,-1}(t,c))$}\\[1mm]
\rmb Moreover, if $Y^*_t(u(x))\partial_x {\bar X_t(x)} $ is Lebesgue-integrable in a neighborhood of $0$, then
\begin{equation}\label{intutility}
U(t,\bar X_t(x))=\int_0^x Y^*_t(u_x(z))\partial_x {\bar X_t(z)} dz,\quad V(t,\zeta^{*}(t,c))=\int_0^c U_x(t,z)d_z\zeta^{*}(t,z)
\end{equation}
Then, with these additional integrability assumptions, $({\bf U}_x, {\bf V}_c)$ are the marginal utilities of a consistent utility system with
consumption.
\end{Theorem}

\begin{proof} The proof of the {\sc main result} is easy, given the previous results.\\
\rmi This assertion is a simple consequence of assumptions on the orthogonal diffusion characteristics.\\
\rmii a) We start by solving the wealth SDE with coefficients $\sigma_t^*(x)$ and $\mu_t^{*,c}(x)$. These coefficients are locally
Lipschitz, with linear growth  since $\zeta^*(t,x)=(V_c)^{-1}(t,U_x(t,x))\le K^3_t x$. Then a strong solution $X^*$ exists
up to a
explosion time $\tau(x)$. But, by verification from the SPDE,  $U_x(t,X_t^*(x))=Y^*_t(u_x)$ on $[0, \tau(x))$. Since
$Y^*_t(u_x)$ is well defined,
$\tau(x)=+\infty \quad a.s.$.\\
The formulation of the {\sc Reverse problem with consumption } is a more complex, since we have to take into account the incresainf function $\zeta$. The assertion \rma is proved by the same argument as before using Theorem \ref{pp: composition formula}. The assertion \rmb
gives an intuitive form to the construction of the forward utility itself by application of the change of variable formula.
\end{proof}

\vspace{-6mm}

%

\subsection{ Value function of backward classical utility maximization problem as consistent progressive utility}\label{transition:classic} 
This subsection points out the similarities and the differences between consistent progressive 
utilities and backward classical value functions, and their corresponding portfolio/consumption optimization problems. \\[-8mm]
\paragraph{Classical portfolio/consumption optimization problem and its conjugate problem}
{ The classic problem of optimizing consumption and terminal wealth  is determined by a  fixed horizon $T_H$ and two deterministic utility functions $u(.)$ and $v(t,.)$  defined up to this horizon.}
Using the same notations as in Section \ref{sec:conso}, the classical optimization problem  is formulated as the following maximization problem,
\begin{eqnarray}\label{pbopticlassic}
\sup_{(\kappa,c) \in  \GX^c  }  \mathbb{E} \Big( u(X^{\kappa,c}_{T_H})+\int_{0}^{T_H} 
v(t,{c^\mathbb{}_t}) dt \Big).
\end{eqnarray}
For any $[0, T_H]$-valued  $\mathbb{F}$-stopping $\tau$  and  for any positive random variable 
$\mathcal{F}_{\tau}$-mesurable $\xi_{\tau}$,  $ \GX^c(\tau,\xi_{\tau})$ denotes the set of admissible strategies  starting at time $\tau$ with an initial positive wealth  $\xi_{\tau}$, stopped when the wealth process reaches 0.
The corresponding value system (that is a family of random variables indexed by $(\tau, \xi_{\tau})$) is defined as, 
\begin{eqnarray}\label{pbopticlassic}
\mathcal{U}(\tau,\xi_{\tau})= \esssup_{(\kappa,c) \in  \GX^c(\tau,\xi_{\tau})  } \mathbb{E}  \Big( u(X^{\kappa,c} _{T_H}(\tau,\xi_{\tau}))+\int_\tau^{T_H}
v(s,{c^\mathbb{}_s}) ds | {\cal
F}_\tau \Big), \>a.s.
\end{eqnarray}
with terminal condition $\mathcal{U}(T_H,x)=u(x)$. \\
{\em We assume the existence of  a  progressive utility (still denoted $\mathcal{U}(t,x)$) aggregating  these system}:  this result  is more or less implicit in the { literature} and  has been proven by Englezos and Karatzas  \cite{KarEng} in the case of a complete market. The proof for an incomplete market will be done in  a future work. \\
 As it is classical in such stochastic control problems (\cite{ElKaroui}) and shown by W. Schachermayer in \cite{MR2014244} for problem without consumption, the dynamic programming principle reads as follows:
 for  any pair $\tau
\leq \vartheta$  of $[0,T_H]$-valued stopping times
\bc
$\mathcal{U}(\tau, \xi_{\tau})=\esssup_{(X,c) \in  \GX^c(\tau, \xi_{\tau})  }\mathbb{E}\Big(\mathcal{U}(\vartheta,X_{\vartheta}(\tau,\xi_{\tau}) +   \int_{\tau}^{\vartheta} v(s,c_s) ds|{\cal F}_\tau\Big)\>a.s.$
\ec
Under mild assumptions on the asymptotic elasticity of utility functions  $(u,v)$, it is also proved in \cite{Kramkov4} and \cite{MR2014244} the existence for any initial wealth of an optimal solution (portfolio, consumption). Then, {\em
$(\mathcal{U}(t,x), v(t,c))$ is a $\GX^c$-consistent dynamic utility system in the sense of Definition \ref{def:conso} up to
time $T_H$. } The same property was proved by Mania and Tevzadze \cite{Mania} in a problem without consumption under
strong  regularity assumption on the value function $\mathcal{U}$ by using backward SDPE.\\
Similarly,  let  $(\widetilde{\mathcal{U}}(t,y), \tilde v(t,y))$ be the  convex conjugate
 utilities  of $(\mathcal{U}(t,x), v(t,c))$. $\widetilde{\mathcal{U}}(t,y)$ aggregates the dynamic version of 
the equivalent backward dual  problem (Karatzas-Lehoczky-Shreve \cite{Karatzas04}) defined, for any $\mathbb{F}$-stopping time $\tau\leq T_H$  and  for any positive random variable $\mathcal{F}_{\tau}$-mesurable $\psi_{\tau}$,
from  the family $\GY^c(\tau, \psi_\tau)$ of  the state price density processes $\{ Y^\nu,  \nu \in \sigoR   \} $ (see
(\ref{Ynu})) with dynamics $dY^\nu_t(y)=Y^\nu_t(y)[-r_tdt+ (\nu_t-\eta^\sR_t).dW_t],\>
\nu_t \in \sigoR_t$, starting from $\psi_{\tau}$ at time $\tau$. The value function of the dual backward optimization problem is then
\begin{eqnarray}\label{dualopti}
 \widetilde{\mathcal{U}}(\tau, \psi_\tau)=  {\rm ess\,inf}_{Y^\nu \in  \GY^c(\tau,\psi_{\tau})  } \mathbb{E} \Big(
\tilde{u}(Y^{ \nu}_{T_H}) + \int_\tau^{T_H} \tilde{v}(s, Y^{ \nu}_{s}) ds  | {\cal
F}_\tau \Big), \>a.s.
\end{eqnarray}
\vspace{-10mm}
\paragraph{Optimal processes}\label{optimalproce}
The process $\b(Y_t^0(y)\b)$ associated with $\nu=0$ is linear in $y$, and  denoted $Y_t^0(y)=yY_t^0(1)=yY_t^0$, for simplicity. We also frequently used the shorthand notation $Y^0_{s,t}=Y^0_{t}/Y^0_{s}, \>s<t$.\\
\rmi  In a complete market, $\b(yY_t^0\b)$ is the unique state price density process,  and the value function is given by 
\bc
$\widetilde{\mathcal{U}}(\tau, \psi_\tau)= \mathbb{E} \Big(
\tilde{u}( \psi_\tau Y^{0}_{\tau, T_H}) + \int_\tau^{T_H} \tilde{v}(s, \psi_\tau Y^{0}_{\tau,s}) ds  | {\cal
F}_\tau \Big), \>\>\P.a.s.$
\ec
Then the optimal state price density $Y^*$ does not depend on the utility functions $u$ et $v$, and on the horizon  $T_H$, to the difference of the optimal processes $(X^*,c^*)$.\\
\rmii In an incomplete market,
we refer to  \cite{Kramkov4}
 to ensure   the existence of an optimal state price density  $Y^{\nu^*,T_H}$ minimizing the criterium
 (\ref{dualopti}) (observe that this can be done under weaker assumptions than in the forward case). Now, the optimal choice depends on the horizon and obviously on the utility functions $(u,v)$. Since we are essentially interested in the horizon dependency, we do not recall the influence of the utility criterium. To avoid too cumbersome  formula, $Y^{\nu^{*,T_H}}$ is often denoted $Y^{*,H}.$\\
Since $(\widetilde U, \tilde v)$ is a progressive conjugate utility system, the results of Theorem \ref{ppU_xconso}, (obtained directly in the { aforementioned } works by the maximum principle), we known that  $U_x(t,X^{*,H}_t(x))$, $Y_t^{*,H}(y)$ and the
 optimal consumption rate $c^{*.H}_t(c_0)$ are linked by their initial conditions:  $\mathcal{U}_x(0,x)=y= v_c(0, c_0)$
 \begin{equation}\label{maxconso}
 \left\{
 \begin{array}{cll}
 c^{*,H}_t(c_0)&= -   \tilde{v}_y(t, Y^{{*,H}}_{t}(y))  \quad   \mbox{ i.e. } \quad   &v_c(t, c^{*,H}_t(y)) =  Y^{*,H}_t(y), \\
{\mathcal{U}_x(t, X^{*,H}_t(x))} &=  Y^{*,H}_t(\mathcal{U}_x(0,x))=v_c(t, c^{*,H}_t(- \tilde{v}_y(0, y))) \quad& \mathcal{U}_x(0,x)=y= v_c(0, c_0)
 \end{array}
 \right .
 \end{equation}
  But 
the way  the classical  optimization   problem is posed is completely different 
from the progressive utility problem: the initial value of $\mathcal{U}$  is computed through a backward analysis, starting from its given terminal value, whereas for  progressive utilities  the initial value is given. Moreover, the optimal wealth is characterized only from its terminal value $X^{*,H}_{T_H}$. To compute its value at any time $t$, we have to use pricing techniques based on the fact that $( X^{*}_{T_H} +\int_0^{T_H} c^*_s ds)$ is a replicable asset, and its market value at time $t\leq T_H$ is given by
\begin{eqnarray}\label{Pricingrule} 
X^{*,H}_t(x) = \mathbb{E} \B( Y^{*,H}_{T_H}(y)   X^{*,H}_{T_H} + \int_t^T Y^{*,H}_s(y)  c^{*,H}_s(c_0) ds | \mathcal{F}_t \B).
 \end{eqnarray}
%
An other major difference is that in the backward point of view no attention is paid to the monotony of  optimal strategies.\\[-8mm]

\paragraph{Example of horizon dependency} In contrast to the progressive utility framework, the optimal solution in the
 classical setting highly depends on the horizon $T_H$, which leads to intertemporality issues. 
To illustrate this
 time-inconsistency, let us consider  an intermediate horizon $T$ between $0$ and $T_H$ and the following two scenarios.
 \bit
 \item[$-$]
 In the first one, the investor computes his optimal strategy for  the horizon $T$ and the utility functions $(u,v)$,
 and then reinvests at time $T$ its  wealth  $X^*_T$ to realize an optimal strategic policy $(X^{*,H}_{T_H}(T, X^*_T), c^{*,H}_t(T, X^*_T))$, 
optimal between the dates $(T,T_H)$ for  the problem with utility functions  $(u^H,v)$.
 \item[$-$]
 In the second one, the investor  computes his optimal strategy, denoted $(\hat{X}_{T_H}, \hat c^H_t)$, 
  directly   for  the horizon $T_H$ and the utility functions $(u^H,v)$. 
  \eit
By uniqueness of preferences, often implicitly assumed by the investors,
 the terminal value of both scenarios must coincide, that is  $X^*_{T_H}(T, X^*_T)=\hat{X}_{T_H}\> a.s.$. 
for any $T$ and $T_H$ ($T<T_H$).  This is impossible in general.  Indeed, between $(T,T_H)$ the investor is using the same utility functions, $(u^H, v)$ applied to different initial wealths at time $T$,  $X^*_T$ for the first strategy, and $\hat{X}^*_T$ for the second strategy, since ${\hat X}_{T_H}(T, {\hat X}_T)=\hat{X}_{T_H}\> a.s.$ In particular, if $\hat{X}$ is monotonic with respect to the initial wealth, the final time consistency can be done if and only if $\hat{X}_T=X^*_T,\>\P-a.s.$. If we are looking for the same property at any time $T$, the wealth process $\hat{X}$ and $X^*$ are the same.
On the other hand, the dynamic programming principle implies that  ${\hat X}_T$ is the optimal wealth for the classical problem with horizon $T$, but stochastic utility $(\mathcal{U}_x(T,x),v)$. In any case, the optimal strategies can not be the same.

Therefore, progressive utilities processes are an alternative to classical utilities functions that gives time-consistency
properties,  and motivate to reconsider problems issued from classical utility framework, with the light of  intertemporal consistency.  Section \ref{sec:taux} focus on the example of long term discount rates and yield curves. 
But before this, as an application of utility maximization, we recall some results on the pricing of contingent claim in finance.
%
%
%
\subsection{Risk neutral pricing and marginal utility (with consumption) indifference pricing}\label{subsec:pricing} 
In the backward point of view, we have found the market value of the optimal wealth, by the so-called pricing rule (\ref{Pricingrule}). This question is  related to a more general issue in finance, that consists in  the pricing of a bounded contingent claim $\zeta_T$, paid at date $T$, 
 $T \leq T_H$. \\[-8mm]
\paragraph{Risk neutral pricing of hedgeable payoffs}
\rmi In the study of optimal state price density in \ref{optimalproce}, we have seen the "universal" rule played by the so-called minimal density process $Y_t^0(y)=yY^0_t$. In particular, since $\sigR_t$ is a vector space, money market strategies $(\kappa \equiv 0)$ are admissible, and $L_t^{0}=e^{\int_0^tr_s ds} Y^{0}_t $ is a local martingale. 
We now assume that $\b(L_t^{0}\b)$ is a uniformly integrable martingale on $[0,T_H]$, which allows us to introduce a minimal, also called risk-neutral, martingale measure,
\bc 
$d\Q=L^0_{T_H}.d\P$ on the $\>\sigma$-field $\>\cal F_{T_H}$. 
\ec
More generally, for any admissible $\nu \in \sigoR_t$, $L_t^\nu(y)=e^{\int_0^tr_s ds} Y_t^{\nu}(y):=L_t^0\,L_t^{\bot,\nu}(y)$ is also a local martingale, product of the martingale $L^0$  and the orthogonal local martingale $L_.^{\bot,\nu}(y).$ So,  $L_.^{\bot,\nu}(y)$ is a $\Q$-local martingale, with $\Q$-expectation smaller than $y$.
When $\E\b(L_{T_H}^{\nu}(y)\b)=y $, then  $L_{T_H}^{\nu}(y)/y $  is the density of a probability measure $\Q^\nu$ with respect to $\P$, and $L_{T_H}^{\bot,\nu}(y)/y$ is the density of $\Q^\nu$ with respect to $\Q.$\\
\rmii In complete market, or more generally in incomplete market without arbitrage opportunity, the market price $p^m(\zeta_T)$ ($p^m$ when it is not ambiguous) of any bounded contingent claim $\zeta_T$ paid at date $T$
that is replicable by an admissible self-financing portfolio is a bounded process $p^m_t$ such that  $Y^{\nu}_t p^m_t$ is a local martingale for any admissible state price density, in particular for $yY^0_t$ and $Y^*_t(y)$. 
Since $L^0$ is a true martingale, and $\zeta_T$ is bounded, $\b(Y^{0}_t p^m_t\b)$ is also a true martingale given by the conditional expectation of its terminal value; this  observation yields to the classical pricing formula (in a complete market) as the minimal risk neutral conditional expectation of the discounted claim between $t$ and $T$, 
 \begin{equation}\label{zcbond}
 p^m_t=\mathbb{E}\big[\frac{Y^{0}_{T}(y)}{Y^{0}_{t}(y)}\zeta_T\b|\mathcal{F}_t\big]=\mathbb{E}^{\Q}\b[e^{-\int_t^Tr_s ds}\zeta_T\b|\mathcal{F}_t\big].
 \end{equation}
 %
 Moreover since for any admissible process $\nu \in \sigoR$,
 $\b(L^{\bot,\nu}_t(y) e^{-\int_t^Tr_s ds} p^m_t\b)$ is also a positive $\Q$-local martingale, and then a $\Q$-supermartingale, the following inequality (with equality if $L^{\bot,\nu}$ is a $\Q$-martingale) holds true
 \begin{equation}\label{ZCincomplete}
 \mathbb{E}\big[\frac{Y^{\nu}_{T}(y)}{Y^{\nu}_{t}(y)} \zeta_T\b|\mathcal{F}_t\big]=\mathbb{E}^{\Q}\b[\frac{L^{\bot,\nu}_{T}(y)}{L^{\bot,\nu}_{t}(y)} e^{-\int_t^Tr_s ds}\zeta_T\b|\mathcal{F}_t\big]\leq p^m_t, \> \>\P-a.s.
 \end{equation}
 The same pricing formula may be used for pricing bounded hedgeable pay-off. {\em The minimal risk-neutral pricing rule gives the maximal seller price for bounded hedgeable contingent claim.} \\
\rmiii In the forward point of view, we know, from the regularity assumption, that the optimal state price $Y^*$ admits the following decomposition $Y^*_t(y)=yY^0_t L_t^{\bot,*}(y)$,  where $L_t^{\bot,*}(y)$ is a $\Q$-uniformly integrable martingale. Then, all the previous inequalities are equalities and in particular, for hedgeable payoff $\zeta_T$,
\bc
$\mathbb{E}\big[\frac{Y^{0}_{T}(y)}{Y^{0}_{t}(y)}\zeta_T\b|\mathcal{F}_t\big]=\mathbb{E}^{\Q}\b[e^{-\int_t^Tr_s ds}\zeta_T\b|\mathcal{F}_t\big]=p^m_t.
 \> \>\P-a.s.$
\ec
 The same property holds true in the backward case, on the assumption that $L_t^{\bot,*,H}(y)$ is a $\Q$-uniformly integrable martingale.\\[-8mm]
\paragraph{Marginal utility indifference pricing} When the payoff $\zeta_{T_H}$ is not replicable 
 in incomplete market,
there are different ways  to evaluate the risk coming from the unhedgeable part,   yielding to a bid-ask spread.  A way is the pricing by indifference.\\
 When the investors  are  aware of their sensitivity to 
the unhedgeable risk, they can try to transact for only a little amount in the risky
contract. In this case, the buyer wants to transact at the buyer's "fair price" (also called {\emph{Davis price} or {\emph{marginal utility price} \cite{Davis}), which corresponds to the zero marginal rate of substitution ${p}^{u}_t$.
In other words, considering the two following backward maximization problems (with and without the claim $\zeta_{T_H}$):
\begin{eqnarray}
\label{claim problem}
\mathcal{U}^{\zeta}(t,x,q)
&
:= \sup_{(\kappa,c) \in\GX^{c} (t,x)} \E[U(T_H,X^\kappa_{T_H}+q \>\zeta_{T_H}) + \int_t^{T_H} V(s,c_s)ds |\F_t],
\\ 
\label{basic primal problem}
\mathcal{U}(t,x)
&
:= \sup_{(\kappa,c) \in\GX^c(t,x)} \E[U(T_H,X^\kappa_{T_H})   +  \int_t^{T_H} V(s,c_s)ds|\F_t],\quad t\leq T_H
\end{eqnarray}
 the marginal utility {\em indifference price} { is the price at which the investor is indifferent from investing or not in the contingent claim}:  it is the  $\F_t$-adapted process $(p^{u}_t(x))_{t\in[0,T_H]}$ determined at any time $t$ by the non linear relationship
\begin{align}\label{abstract indifference price problem}
\partial_q\mathcal{U}^{\zeta}(t,x,q)|_{q=0}=\partial_q\mathcal{U}(t,x+q p^{u}_t(x)|_{q=0},\quad \text{for all }t\in[0,T_H].
\end{align}
The marginal utility price is a linear pricing rule.  Using this pricing rule
 means that there exists a consensus { on} this price for a small amount, but investors are not sure to have liquidity at this price. 
{ In  the backward case, the marginal utility price, such as the optimal state price density,  depends on the horizon $T_H$.  
 In particular,  if  the contingent claim  $\zeta_T$  is delivered at time $T < T_H$, then  $\zeta_T$ can be invested  between time $T$ and $T_H$ into any admissible portfolio $X_.(T,\zeta_{T} )$ (martingale under $Y^*$) and computing the marginal utility price  with terminal payoff $\zeta_{T_H}=X_{T_H}(T,\zeta_{T} )$ leads to the same  price,  as explained in the following proposition. 
 When needed, we use the notation $Y^{*,H}_.$ and $p^{u,H}_0(x, \zeta_{T})$ to emphasize the time horizon dependency of the backward optimization problem}.

\begin{Proposition}
Let $(U,V)$ be  the progressive utilities associated with a consumption consistent  optimization problem with optimal state price density process $Y^{*}_.(y)$.\\
\rmi For any non negative contingent claim $\zeta_{T_H}$ delivered at time $T_H$, the marginal utility  price (also called Davis-price) is given via the dual parametrization $y$
\begin{equation}\label{Davis price}
p^u_t(x, \zeta_{T_H}) =\E \b[\zeta_{T_H}Y^{*}_{T_H}(t,y)/y|\F_t\b],  \quad y=\mathcal{U}_x(t,x).
\end{equation}
\rmii In the forward case, the pricing rule may be defined for any maturity $T\leq T_H$ in the same way. Then, the pricing rule is time-consistent,$$\> p^u_t(x, \zeta_{T_H})= p^u_t(x, \zeta_T(t,x))    \quad \text{where}      \>  \>       \zeta_T(t,x)=p^u_T(X_{T}^*(t,x), \zeta_{T_H}).$$
\rmiii In  the backward case, the marginal utility {\em indifference price} is only defined for cash-flow paid at horizon $T_H$.  When the claim  $\zeta_{T}$ is delivered at time $T$ before $T_H$,  $\zeta_{T}$ may be considered as the (indifference) price at $T$ of any admissible portfolio starting from $\zeta_{T}$ at $T$ with terminal wealth $X_{T_H}(T, \zeta_T)=\zeta_{T_H}$. 
 The marginal utility price of $\zeta_{T}$, denoted $p^{u,H}_0(x, \zeta_{T})$ to recall its dependency in $T_H$ is then,
 \begin{equation}\label{Davis price at T}
p^{u,H}_t(x, \zeta_{T})=p^{u,H}_t(x, \zeta_{T_H})=\E \b[ \zeta_{T_H}Y^{*,H}_{T_H}(t,y)/y|\F_t\b]
=\E \b[\zeta_{T}Y^{*,H}_{T}(t,y)/y|\F_t\b].
  \end{equation}
\rmiv The backward marginal utility  pricing is a well-posed pricing rule, since it  is not depending on the choice of the admissible extension on $\zeta_{T}$. Moreover, the  rule is also time-consistent.
\end{Proposition}

\begin{proof}
 Following Davis \cite{Davis}, we compute the marginal indifference price of any contingent claim as follows. Denote by $(X^{*,q}(z),c^{*,q}(z))$ the optimal strategy of the optimization program (\ref{claim problem}) ($q$ quantity
of claim $\zeta_{T_H}$), i.e. 
\bc
$\E\b[U(T_H,X^{*,q}_{T_H}(x)+q \zeta_{T_H}) + \int_t^{T_H} V(s,c^{*,q}_s(x))ds\b]=\mathcal{U}^{\zeta}(t,x,q)$. 
\ec
Formally, we can derive with respect to $q$ under the expectation, and take the value of the derivative at $q=0$ (known as the  envelope theorem  in economics)
\begin{eqnarray}
\partial_q\,\mathcal{U}^{\zeta}(0,x,q)|_{q=0}&=\E\B[
\b(U_x(T_H,X^{*,q}_{T_H}(x)) (\partial_q\, X^{*,q}_{T_H}(x)+ \zeta_{T_H})|_{q=0}\b)\nonumber\\ 
&+ \int_0^{T_H} \b(V_c(s,c^{*,q}_s(x))\partial_q\, c^{*,q}_{s}(x)|_{q=0}\b)ds\B].
\end{eqnarray}
Under regularity assumption, it is shown in \cite{Davis} that the optimal  processes ($X_{T_H}^{*,q}, c^{*,q}_s(x)$) are continuously
differentiable  with respect to the quantity $q$ satisfying $\lim\limits_{q\rightarrow 0}X_{T_H}^{*,q}=X_{T_H}^{*}, $
and $\lim\limits_{q\rightarrow 0}\partial_qX_{T_H}^{*,q}=0,  a.s.$;    $\lim\limits_{q\rightarrow 0}c^{*,q}_s(x)=c^{*}_s(x)$
and $\lim\limits_{q\rightarrow 0}\partial_qc^{*,q}_s(x)=0,  a.s..$ \\
This implies that the marginal indifference price satisfies\\[1mm]
 \centerline{$ \E\b[\zeta_{T_H} U_x(T_H,X_{T_H}^{*}(x))\b] 
= p^{u}_0(x)\> \mathcal{U}_x(0,x).$}\\[1mm]
In the forward and backward case, the marginal utility of the optimal wealth at the horizon $T_H$,  $U_x(T_H,X_{T_H}^*(x))$, is the optimal state price density $Y^{*}_{T_H}(y)$ with initial condition $y=\mathcal{U}_x(0,x)$.

The main difference is that in the forward case, the process $Y^{*}$ does not depend of $T_H$ in contrast to the backward setting.  In the forward case,
\begin{equation}\label{Davis price}
p^u_0(x)=  \frac{1}{  \mathcal{U}_x(0,x)}\E\b[U_x(T_H,X_{T_H}^*(x))\> \zeta_{T_H}\b]=\E \b[\zeta_{T_H}Y^{*}_{T_H}(y)/y\b].
\end{equation}

In the backward case, if the maturity of the claim is $T\leq T_H$, then investing the amount $\zeta_T$ in any admissible portfolio $X_.(T,\zeta_T)$ such that $(X_t(T,\zeta_T)\>Y_t^{*,H}(y))_{t\ge T}$ is a martingale and  taking $\zeta_{T_H}=X_{T_H}(T,\zeta_T) $, it follows that in any case
\\[-6mm]
\begin{eqnarray}\label{capitalisation}
 p^{u,H}_0(x,\zeta_T) &=& p^{u,H}_0(x,\zeta_{T_H}) =  \E\b[ \E( X_{T_H}(T,\zeta_T)    Y^{*,H}_{T_H}(y)/y  |\mathcal{F}_T) \b] \nonumber \\
&= & \E\b[\zeta_{T}Y^{*,H}_{T}(y)/y\b], \quad \quad  y=\mathcal{U}_x(0,x)
\end{eqnarray}
which proves that the  backward marginal utility  pricing is a well-posed pricing rule.
\end{proof}
The same argument may be used at any date $t$ to define the marginal utility price, using the conditional distribution with respect to the filtration $\F_t$,
\begin{eqnarray*}
p^{u}_t(z)&=  \frac{1}{  \mathcal{U}_x(t,z)}\E\b[U_x(T_H,X_{T_H}^*(z))\> \zeta_{T_H}\b|\F_t\b], &\quad z=X_t^*(x)\\
&=\E \b[\zeta_{T_H}Y^{*}_{T_H}(\phi_t)/\phi_t\b|\F_t\b], &\quad \phi=\mathcal{U}_x(t,X_t^*(x))=Y^{*}_t(\mathcal{U}_x(0,x)).
\end{eqnarray*}

\section{Application to  yield curves dynamics}\label{sec:taux}
 For  financing of ecological projects reducing global warming, for
longevity issues or any other investment with a long term impact, it is necessary to model accurately long run interest rates. The answer cannot find in financial market, since for longer maturities (30 years and more), the bond market becomes highly illiquid and standard financial interest rates models cannot be easily extended. 
\subsection{General macroeconomics consideration}
In general, these issues are  addressed at macroeconomic level, where  long-run  interest rates has not necessary the same meaning than in financial market.They are  called {\em socially efficient or economic} interest rates, 
because they would be only affected by structural characteristics of the economy, and to be low-sensitive to monetary policy. Nevertheless, correct estimates of these rates are therefore useful for long term decision making, and understanding their determinants is important. \\[-10mm]
\paragraph{Ramsey rule and equilibrium interest rates} The macroeconomics literature typically relates the {\em economic equilibrium} rate to the time preference rate and to the average rate of productivity growth. 
A typical example is the Ramsey rule proposed in the seminal paper of Ramsey \cite{Ramsey}  in 1928  where  economic interest rates  were linked with the marginal utility of the aggregate consumption at the economic equilibrium.
More precisely, the economy is represented by the strategy of a risk-averse representative agent, whose utility function on 
consumption rate at date $t$ is the function $v(t,c)$. 
 Using an equilibrium point of view, the Ramsey rule  at time $0$ connects the equilibrium rate for maturity $T$ with the
marginal utility $v_c(t,c)$ of the random optimal consumption rate  $(c^*_t)$ by
\begin{equation}\label{RamseyRule}
R^e_0(T)=-\frac{1}{T}\ln\frac{\E[v_c(T,c^*_T)]}{v_c(0, c_0)}.
\end{equation}
An usual setting is to assume separable in time  utility function with exponential decay at rate $\beta>0$ and constant risk aversion $\alpha, (0<\alpha < 1)$, that is $v(t,c)=Ke^{-\beta t}\frac{c^{1-\alpha}}{1-\alpha}$.  $\beta$ is the pure time preference parameter, i.e. $\beta$ quantifies the agent preference of immediate goods versus future ones. The optimal consumption rate is then exogeneous modelled as a geometric Brownian motion, $c^*_t=c_0 \exp((g-\frac{1}{2}\sigma^2) t + \sigma W_t)$ with $g$ the growth rate of the economy. The Ramsey rule induces a flat curve
\begin{equation}\label{RamseyRuleflat}
R^e_0(T)=\beta+\alpha g -\frac{1}{2} \alpha (\alpha+1)   \sigma^2. 
\end{equation}
The Ramsey rule is still the reference equation even if the framework in consideration is more realistic, as its is  was
discussed by numerous economists, such as Gollier
\cite{Gollier6,Gollier9,Gollier13,Gollier14,Gollier15,Gollier16,GollierEcological} and Weitzman \cite{Weitzman,
Weitzman_review}.  
The equilibrium yield curve at time  $0$ is then computed through the Ramsey rule, using the maximum principle and  leaving undiscussed the time-consistency of such an approach.

   Dynamic utility functions seem to be well adapted for modeling and studying long term 
yield curves and their dynamics, because it allows
to get rid of the dependency on the maturity $T_H$ of the classical backward optimization problem and thus gives time consistency for the optimal choices.  
Besides, as dynamic utility functions take into account that the preferences and risk aversion of investor may
change with time, they  are also more accurate.
Indeed,  in the presence of generalized long term uncertainty, the decision
scheme must evolve: the economists agree on the necessity of a
sequential decision scheme that allows to revise the first decisions
according to the evolution of the knowledge and to direct experiences, see Lecocq and Hourcade \cite{Hourcade}. 
Besides, a  sequential decision allows to cope with situations in which it is important to find the core of an  agreement
between partners having different views or anticipations,  in order to give time for solving their controversy. \\[-8mm]

\subsection{The financial framework}
 Cox-Ingersoll-Ross \cite{CIR} adopt an equilibrium approach to endogenously determine the
term structure of interest rates, in the presence of a financial market. In their model, there exists a single consumption good
and the production process follows a diffusion whose coefficients depends on an exogeneous
stochastic factor which in some way influences the economy.
The risk-free rate is determined endogenously such that the investor is not better off by
trading in the money market, i.e. she is indifferent between an investment in the production
opportunity and the risk-free instrument.\\
The financial point of view presented now is very closed to the previous one, but the agent may  invest in a financial market in addition to the money market.
{We consider an arbitrage approach with exogenously given interest rate,  instead of an equilibrium approach that determines them endogenously (see the Lecture notes of  Bj\"ork \cite{bjork} for a comparison between these two approaches). The financial market is 
 an incomplete It\^o  financial  market: notations are  the one described in Section \ref{descmarket},  with a $n$ standard Brownian
motion $W$, a (exogeneous) financial short term interest rate  $(r_t)$ and a $n$-dimensional risk premium $(\eta^\sR_t)$.} \\[-8mm]

\paragraph{The  (backward) classical optimization problem}
  In the  classical optimization  problem  (\ref{pbopticlassic}) with given horizon $T_H$, studied in Subsection \ref{transition:classic}, both utility functions for terminal wealth and consumption rate are deterministic, then designed by small letter $u$ and $v$; their Fenchel conjugates are denoted by $(\tilde{u}, \tilde{v})$. \\ Since we are concerned essentially  by the Ramsey rule and the yield curve dynamics, we focus on the equivalent dual formulation 
 (\ref{dualopti}).  
The 
 optimal consumption rate $c^{*,H}(y)$ depends on the time horizon $T_H$ through the optimal state price density process $Y^{{*,H}}$ 
 \begin{equation*}(\ref{maxconso})
 \left\{
 \begin{array}{cll}
 c^{*,H}_t(c_0)&= -   \tilde{v}_y(t, Y^{{*,H}}_{t}(y))  \quad   \mbox{ i.e. } \quad   &v_c(t, c^{*,H}_t(y)) =  Y^{*,H}_t(y), \quad  0\leq t\leq {T_H}\\
c_0&=-   \tilde{v}_y(0, y)\quad \mbox{ i.e. } & v_c(0, c_0) = y
 \end{array}
 \right .
 \end{equation*}
As the Lagrange multiplier $y$ does  not have  an obvious financial interpretation, we adopt as in the economic literature the parameterization by the initial consumption $c_0$, based on  the one to one correspondance ${v_c(0, c_0)} =  y. $ \\ 
Equation (\ref{maxconso}) may be interpreted as a pathwise Ramsey rule, between  the marginal utility of the  optimal  consumption  and the optimal state price density process:

\begin{equation}\label{Ramseypath}
 \frac{v_c(t, c^{*,H}_t(c_0))}{v_c(0, c_0)} =  \frac{Y^{*,H}_t(y)}{y},   \quad    0\leq t\leq {T_H}       \quad   \mbox{ with  } \quad {v_c(0, c_0)} =  y.
\end{equation}


\paragraph{The (forward point of view) dynamic problem}
 We adopt  notations of Section \ref{sec:conso}, using capital letter to refer to progressive utilities.
The pathwise relation  (\ref{maxconso}) still holds for progressive  utility functions, using the  characterization of the optimal consumption (see Theorem \ref{ppU_xconso}), where the parameterization is done through the initial wealth $x$, or equivalently $c_0$ or $y$ since
$c_0 =  -\tv_y(u_x(x))=  -\tv_y(y) $,
\begin{equation}\label{Ramseypathdynamic}
V_c(t, c^{*}_t(c_0)) = Y^{*}_t(y),   \quad    t \geq 0     \quad   \mbox{ with  } \quad v_c(0, c_0) =  y.
 \end{equation}
 The forward point of view emphazises the key rule played by the monotony of $Y$ with respect to the initial condition $y$, under 
regularity conditions of the progressive utilies (cf Theorem \ref{ppU_xconso}). 
Then as function of $y$, $c_0$ is decreasing, and $c^{*}_t(c_0)$ is an increasing function of $c_0$.
This question of monotony is frequently avoided, maybe because with power utility functions (the example often used in the litterature)
$Y^{*}_t(y)$ is linear in $y$  as $\nu^*$ does not dependent on $y$.\\
The optimal state price density process $Y^{*}$ summarizes all the difference between the classical backward and dynamic forward approachs. 
In particular, progressive utilities allows
to get rid of the dependency on the maturity $T_H$ and thus gives time consistency 
for the optimal choices.\\
{\bf Remark}:    This time unconsistency is also present in the  Ramsey rule (\ref{RamseyRule}) in the economic litterature. Indeed, the optimization problem they considered is usually formulated 
through a time separable utility  
  $v(t,x)=e^{-\beta t}v(x)$  with a infinite horizon, which is equivalent\footnote{ If
$\tau_H$ is  distributed  as an independent exponential law  with parameter $\beta$,\\
 $ \hspace*{2cm}  \mathbb{E}(\int_0^{+\infty } e^{-\beta t} v(c_t)dt)= \mathbb{E}(\int_0^{\tau_H }  v(c_t)dt)   $.  } (in expectation) 
 to consider the utility $v$ and a random horizon $\tau_H$ exponentially distributed   with parameter $\beta$.
In the Ramsey rule (\ref{RamseyRule}), the optimal consumption process $c^*$ intrinsically depends on $\beta$, 
which corresponds to the dependency on the horizon $T_H$ of our classical backward formulation.\\[-8mm]
\subsection{Equilibrium and financial yields curve dynamics}
As previously observed, forward and backward optimization problems lead to the same pathwise  relation
(\ref{Ramseypathdynamic}) 
between optimal consumption and optimal state price density. The main difference is in the dependence on the horizon of optimal quantities in the backward case. {So, in general the notation of the forward case are used, but
 with the additional symbol $H$ ($Y^{*,H}, c^{*,H}, X^{*,H}$) to address the dependency
 on $T_H$ in  the classical backward problem. }\\
%
%
%
\rmi Thanks to the pathwise relation (\ref{Ramseypathdynamic}), the Ramsey rule yields to a description of the equilibrium interest rate as a	 function of the optimal  state price density process $Y^*$, $R^e_0(T)(y)=-\frac{1}{T}\ln \mathbb{E}[Y^{*}_T(y)/y]$, that allows to give a financial interpretation in terms of zero coupon bonds.
More dynamically in time, 
\begin{equation}\label{Ramseyfin}
R^e_t(T)(y):=-\frac{1}{T-t} \ln \mathbb{E}\left[\frac{V_c(T,c^*_T(c^*_0))}{V_c(t,c^*_t(c_0))}\big| \mathcal{F}_t\right]=-\frac{1}{T-t}\ln
\mathbb{E}\left[\frac{Y^{*}_T(y)}{Y^{*}_t(y)}\b| \mathcal{F}_t\right]\quad \forall t < T.
\end{equation}
Thanks to the flow property, $\{Y^{*}_T(y)=Y^{*}_T(Y^{*}_t(y)), c^*_T(c^*_0)=c^*_T(c^*_t(c^*_0)), \>t<T\}$,  the equilibrium yield curve  starting at time $t$ with initial condition $c^*_t(c_0)=-\tV_y(t,Y^{*}_t(y)) $ is given by $\b(R^e_t(T)(Y^{*}_t(y)),t<T\b)$. \\
%
%
\rmii The question is reduced to give a financial interpretation in terms of price of zero-coupon bonds, of the quantities $\mathbb{E}\Big[\frac{Y^{*}_T(y)}{Y^{*}_t(y)}\b| \mathcal{F}_t\Big]\quad \forall t<T$.
Let  $\b(B^m(t,T),\>t\leq T\b)$, ($m$ for market), be the price  at time $t$ of a  zero-coupon bond paying one unit of cash at maturity $T$.
In finance, the market yield curve is defined through the price of zero-coupon bond by $B^m(t,T)=\exp(-R^{m}_t(T)(T-t)).$
 We use the  results of Subsection \ref{subsec:pricing} concerning the pricing of contingent claim: the case of zero-coupon bond $B^m(t,T)$ corresponds to $\zeta_T=1$.\\[-8mm]
 \paragraph{Marginal utility yield curve}
\rmi In a complete market, or if the  zero-coupon bonds are hedgeable,   $B^m(t,T)$ is computed by the the minimal risk neutral pricing rule
 $B^m(t,T)=\mathbb{E}\big[Y^0_{t,T}\b|\mathcal{F}_t\big]=\mathbb{E}^{\Q}\b[e^{-\int_t^Tr_s ds}\b|\mathcal{F}_t\big].$\\
{\em Then, for replicable bond, equilibrium interest rate and market interest rate coincide}.\\
 \rmii For non hedgeable zero-coupon bond, we can apply the marginal indifference pricing rule (with consumption). So we denote by $B^{u}(t,T)$ ($u$ for utility) the marginal utility price at time $t$ of  a zero-coupon  bond paying one cash unit at maturity $T$, that is $\quad B^{u}(t,T)=B^{u}_t(T, y)=  \mathbb{E}\B[ \frac{Y^{*}_{T}(y)}{Y^*_t(y)} \b| \mathcal{F}_t  \B]$. Based on the link between optimal state price density and optimal consumption, we see that
\begin{equation}\label{linkYyieldcurve}
B^{u}_t(T,y):=B^{u}(t,T)(y)=   \mathbb{E}\B[ \frac{Y^{*}_{T}(y)}{Y^*_t(y)} \b| \mathcal{F}_t  \B]=\mathbb{E} \B[\frac{V_c(T,c^{*}_T(c_0))}{V_c(t,c^*_t(c_0))} \b| \mathcal{F}_t\B].
\end{equation}
{\em According to the Ramsey rule (\ref{Ramseyfin}), equilibrium interest rates and marginal utility interest rates are the same. Nevertheless, this last curve is robust only for small trades .}\\
The martingale property of $Y^*_t(y) B^{u}_t(T,y)$  yields to the following dynamics  for the zero coupon bond maturing at time $T$ with  volatility vector $\Gamma_t(T,y)$ 
\begin{equation}
\frac{dB^{u}_t(T,y)}{B^{u}_t(T,y)}= r_t dt+\Gamma_t(T,y).(dW_t+(\eta^\sigR_t-\nu^*_t(y))dt).
\end{equation}
Using the classical notation for exponential martingale, ${\cal E}_t(\theta)=\exp\b(\int_0^t \theta_s.dW_s-\demi \int_0^t\|\theta_s\|^2.ds\b)$,  the martingale $Y^*_t(y) B^{u}_t(T,y)$  can written as an exponential martingale with volatility $\b(\nu_.^*(y)-\eta^\sigR_.+\Gamma_.(T,y)\b)$. In particular, 
using that $B^{u}_T(T,y)=1$,
\begin{eqnarray*}
Y^*_T(y)= B^{u}_0(T, y){\cal E}_T\b(\nu_.^*(y)-\eta^\sigR_.+\Gamma_.({\bf T},y)\b)= y\> e^{-\int_0^T r_s ds} \cal{E}_T\b(\nu_.^*(y)-\eta^\sigR_.\b).
\end{eqnarray*}
Taking the logarithm gives
\begin{equation}\label{eq:intrsds}
\int_0^Tr_s ds=TR^u_0(T)-\int_0^T\Gamma_t(T,y).(dW_t+(\eta_t-\nu^*_t(y))dt)+\demi \|\Gamma_t(T,y)\|^2 dt .
\end{equation}
When the family $\Gamma_t(T,y)$ is assumed 
to be differentiable with respect to the maturity $T$, we recover  the classical Heath Jarrow Morton framework \cite{HJM} with the following dynamics representation of the short rate
\begin{equation}
r_t=f_0(t, y)-\int_0^t\partial_T\Gamma_s(t,y).(dW_s+(\eta_s-\nu^*_s(y))ds)+\demi \partial_t\|\Gamma_s(t,y))\|^2 ds
\end{equation}
with $f_0(.,y)$ being the forward short rate.\\

%
\noindent $\bullet$  {\em  Yield curve dynamics and infinite maturity} \\
 The computation of the marginal utility  price of zero coupon bond is then straightforward using (\ref{linkYyieldcurve})
leading to the yield curve dynamics $(R^{u}_t(T,y)= - \frac{1}{T-t} \ln B^{u}_t(T,y))$
\begin{eqnarray*}
R^{u}_t(T,y)& =&\frac{T}{T-t}
R^{u}_0(T,y)  -  \frac{1}{T-t} \int_0^t r_s ds - \int_0^t  \frac{\Gamma_s(T,y)}{T-t} dW_s  \\
 &&+ \int_0^t \frac{||\Gamma_s(T,y) ||^2}{2(T-t)} ds +  \int_0^t  <\frac{\Gamma_s(T,y)}{T-t}, \nu_s^{*}-\eta^\sigR_s   >    ds.
 \end{eqnarray*}
Along the same lines as in Dybvig \cite{Dybvig} and  in El Karoui and alii. \cite{ElKarouiFrachot}, we study the dynamics behavior of the yield curve for infinite maturity, when the maturity goes to infinity
\begin{equation}\label{ltinfini}
l_t(y):= \lim_{T \rightarrow + \infty}R^{u}_t(T,y). 
\end{equation}
If $\lim_{T \rightarrow + \infty}  \frac{\Gamma_t(T,y)}{T-t}$ is not equal to zero $dt \otimes d\mathbb{P}$ a.s. then 
$\lim_{T \rightarrow + \infty}  \frac{||\Gamma_t(T,y)||^2}{T-t}= +\infty$ a.s and $l_t(y)$ is infinite.
Otherwise, 
$
l_t=l_0+ \int_0^t \lim_{T \rightarrow + \infty}  \left( \frac{||\Gamma_s(T,y)||^2}{2(T-s)}   \right) ds
$
thus $l_t$ is constant if $\lim_{T \rightarrow + \infty}  \frac{||\Gamma_t(T,y)||^2}{T-t}=0$ and $l_t$ is  a non-decreasing process if $\lim_{T \rightarrow + \infty}
 \frac{||\Gamma_t(T,y)||^2}{T-t}>0$.


\begin{Remark}{\rm
When hedging strategies cannot be implemented, the nominal amount of the transactions becomes an important risk factor
and marginal utility prices are not accurate any more, especially when the market is highly
illiquid. To face this problem, the utility based indifference pricing methodology  seems to be more appropriate. \\
The \emph{utility indifference price} is the cash amount $\widehat p$ for which the investor is indifferent between selling or buying a certain quantity $q$ of a positive claim $\zeta_{T_H}$ (paid at time $T_H$) at the price $\widehat p$  in an optimally managed portfolio with initial wealth $x+p$ or investing optimally its initial wealth in the market without the claim $\zeta_{T_H}.$  If $q>0$ (resp. $q<0$)  $-p=:p^b$ is a positive buying (resp. $p=:p^s$ is a  selling) indifference price.
In other words, considering the two  backward maximization problems recalled in \eqref{claim problem}:
\begin{align}
\mathcal{U}^{\zeta}(t,x+\widehat p_t,q)=\mathcal{U}(t,x),\quad \text{for all }t\in[0,T_H].
\end{align}
The pricing rule is now non linear, providing a bid-ask spread.  Since it is not possible to develop this idea here, we refer the interested reader to the book "Indifference Pricing" edited by Carmona \cite{Carmona}.}
\end{Remark}

\subsection{Power utilities with consumption and yield curve properties}
To be able to give more precise properties of the marginal utility yield curve, we study  progressive and backward power utilities as the classical most important example
for economics, due to the simplification of some calculation.\\[-8mm]
\paragraph*{Consumption consistent progressive power utility} 
  Let us consider a consumption consistent progressive power utility (with risk aversion coefficient $\alpha$), associated with the pair of power progressive utilities $\b(U^{(\alpha)}(t,x)=\widehat Z^{(\alpha)}_t \frac{x^{1-\alpha}}{1-\alpha}\>, V^{(\alpha)}(t,x)=(\hat \psi_t)^\alpha U^{(\alpha)}(t,x)\b)$. From Corollary \ref{consforwardpower}, the  optimal processes are linear with respect of their initial condition, i.e. $\widehat X^*_t(x)=x \widehat X^*_t, \> Y^*_t(y)=y Y^*_t, \>\text{and}\quad c^*_t(z)=z\>\widehat \psi_t>0.$
The coefficient  $\widehat Z^{(\alpha)}_t$ is determined by the optimal processes via $\widehat Z^{(\alpha)}_t=Y^*_t (\widehat X^*_t)^\alpha$. Moreover, 
$d\widehat X^*_t=\widehat X^*_t\b((r_t-\widehat \psi_t)dt+ \kappa_t^*.(dW_t+\eta^\sigR_t)\b)$ and $dY^*_t=Y^*_t\b(-r_t dt +(\nu^*_t-\eta^\sigR_t)dW_s\b)$
where only the dynamics of $\widehat X^*_t$ is affected by the consumption rate $\widehat \psi_t.$

\paragraph*{Power Backward Utilities and Yields curve}
 \rmi { For backward utility function, the time horizon $T_H$ plays a crucial. For  zero coupon with maturity $T<T_H$, the payoff at time $T$ is recapitalized at the risk-free rate from time $T$ to time $T_H$, leading to the marginal utility  price for zero-coupon (as explained in  (\ref{capitalisation}))
$B^{u,H}_t(T,y)=   \mathbb{E}\B[ \frac{Y^{*,H}_{T}(y)}{Y^{*,H}_t(y)} \b| \mathcal{F}_t  \B]$.}
Since the value function of a power backward utility problem is a consistent power utility process with  deterministic value
at maturity $T_H$, the previous system (I) in   Section \ref{sec:expuissance} states that 
we are looking  for optimal processes $X^*$ and $Y^*$ such that $Z^*_{T_H}=Y^*_{T_H}(X^*_{T_H})^\alpha$ is a constant $C$. { Thus, compared to the forward case, the dependency on the time horizon $T_H$ adds a deterministic constraint between optimal wealth and optimal dual process at time $T_H$}.  
This  constraint is  equivalent to the martingale property of the value function along the optimal portfolio, equal to the martingale 
$\frac{1}{1-\alpha}Y^*_tX^*_{t}=\frac{1}{1-\alpha}{\cal E}_t\b(\kappa^*-\eta^\sigR+\nu^*\b)$.
To understand the impact of the short rate uncertainty, it is better to write the constraint as \\[1mm]
\centerline{$X^*_{T_H}Y^0_{T_H}=K\>Y^0_{T_H}\b(Y^*_{T_H}X^*_{T_H}\b)^{1/1-\alpha}$}\\[2mm]
 since  both  processes
$X^* Y^0$ and $Y^*X^*$ are exponential martingales with known volatility given respectively by $\kappa^*+\eta^\sigR$ and $\nu^*+\kappa^*-\eta^\sigR$, and $Y^0_{t}=\exp(-\int_0^tr_s ds) L^0_{t}$ where  $L^0_{t}$ is an exponential martingale with volatility $-\eta^\sigR$. To characterise the parameters of all these processes, we can use the uniqueness of the decomposition as terminal value of some exponential martingale,  after taking into account the randomness of spot rate $r$ or risk premia $\eta^\sigR$, and $\nu^*$. In any case, this condition implies some links  on the random variable $\int_0^{T_H}r_s ds$ and the volatilities $\nu^*_t-\eta^\sigR_t$ and $\kappa^*_t$ of the optimal processes $Y^*$ and $X^*$. But it is not so easy to give a description of this links in all generality. 
%
%
\paragraph*{Examples in log-normal market of  marginal  utility yields curves with backward power utilities}
 We assume a log-normal market:\\
\rmi  $\eta^\sigR_.$ is a deterministic process (and $\sigR_t$ contains the deterministic processes)\\
 \rmii  $ (\int_0^{t} r_sds)_{0\leq t \leq T_H}$ is a Gaussian process, with a deterministic volatility vector $\Gamma_.(t)$, Thus the logarithm of $Y^0$ is a Gaussian process and equation (\ref{eq:intrsds}) can be written as 
\begin{equation}\label{dynamicsrt}
-\int_0^{t} r_sds= \text{Cst}(t)+\int_0^{t} \Gamma_s(t).dW_s, \quad t \in[0,T_H].
\end{equation}
\rmiii Assuming furthermore that $\nu^{*,H}$ is deterministic, the  logarithm of the    optimal wealth $\ln(X^{*,H})$ and of the optimal state price density $\ln(Y^{*,H})$ are Gaussian process.\\}
In particular, at time $T_H$
\begin{eqnarray*}
\ln(Y^{*,H}_{T_H})&=& \text{Cst}-\int_0^{T_H} r_tdt +\int_0^{T_H}(\nu^{*,H}_t-\eta^\sigR_t).dW_t\label{dynamicsY}\\
 \ln(X^{*,H}_{T_H})&=&\text{Cst}+\int_0^{T_H} r_tdt+\int_0^{T_H} \kappa^*_tdW_t, \nonumber
\end{eqnarray*}
and since  $Y^{*,H}_{T_H}(X^{*,H}_{T_H})^\alpha$ is a constant,  
the Gaussian variable $(1-\alpha)\int_0^{T_H}\Gamma_t({T_H}).dW_t+\int_0^{T_H}(\nu^{*,H}_t-\eta^\sigR_t)dW_t+\int_0^{T_H} \alpha \kappa_t^{*,H} dW_t$ has $0$ variance. Thus, using the decomposition of $\Gamma_t({T_H})$ into two orthogonal vectors $\Gamma^\sigR_t({T_H})$ and $\Gamma^{\perp}_t({T_H})$, we have that
\begin{equation}\label{eqvolbond}
\nu^{*,H}_t=-(1-\alpha)\Gamma^{\perp}_t({T_H}),\quad \alpha \kappa_t^{*,H} +(1-\alpha)\Gamma^\sigR_t({T_H})=\eta^\sigR_t.
\end{equation}
Remark that $\nu^{*,H}$ is always proportional to $\Gamma^{\perp}({T_H})$ and $\kappa^{*,H}$ depends on the maturity $T_H$ only through $\Gamma^\sigR({T_H})$.
So the knowledge of deterministic risk premium $\eta^\sigR$, and  the optimal deterministic parameters $\nu^{*,H}_t$,  $\kappa_t^{*,H}$ allows us to identify the volatility of the marginal utility zero-coupon bond with maturity $T_H$, as 
\begin{equation}
\Gamma_t(T_H)=\frac{(\eta^\sigR_t-\nu^{*,H}_t)}{(1-\alpha)}-\frac{\alpha}{1-\alpha}\kappa_t^{*,H}
\end{equation}
Conversely, given a deterministic volatility for the zero-coupon bond  with maturity $T_H$, and the risk aversion coefficient $\alpha$, we can easily recover from Equation \eqref{eqvolbond} the optimal volatilities $\nu^{*,H}_t$ and $\kappa_t^{*,H}$. \\

\noindent {A classical model for the short rate dynamics is the Vasicek model, where  the short rate is 
 given by an Ornstein-Uhlenbeck process 
$dr_t=a(b-r_t)dt- \sigma dW_t.$
The computation of $\int_0^t r_s ds$ yields the volatility for the zeron-coupon bond   $\Gamma_s(t) =    (1-e^{-a(t-s)})  \frac{\sigma}{a}$ (see for example \cite{JYC} Proposition
2.6.1.6 for details of this classical Gaussian computation).
The classical framework consists in a complete market driven by a one dimensional Brownian motion. In the framework of an incomplete market with the noise driving the spot rate  being orthogonal to the one driving the risky assets, then $\Gamma^{\perp}_s(t) =    (1-e^{-a(t-s)})  \frac{\sigma}{a}$ and $ \Gamma^\sigR=0 .  $
 Thus  in this example $ \kappa_t^{*,H}= \frac{\eta^\sigR_t}{\alpha} $ does not depend on the maturity $T_H$ while
$ \nu_t^{*,H}=  (\alpha-1) (1-e^{-a(T_H-t)}) \frac{\sigma}{a}  $ depends on the time to maturity $(T_H-t)$.\\

}

\noindent $\bullet$ {\em Yield curve for infinite maturity} 
Since in the backward case  $\nu^{*,H}$  depends on the maturity $T_H$, the yield curve for infinite maturity  
$l_t = \lim_{T \rightarrow + \infty}R^{u}_t(T)$ differs  from the one in the forward case (in (\ref{ltinfini})) if $\lim_{T \rightarrow + \infty}
 \frac{||\Gamma_t(T)||^2}{T-t}>0$ and $\alpha < \frac{1}{2}$.
As we are looking at the asymptotics $T \rightarrow + \infty$ and $T \leq T_H$, 
 we set   $T_H=T \rightarrow + \infty  $ (note that similar results hold if  $T_H> T \rightarrow + \infty$) and
\begin{eqnarray*}
l_t&=&l_0+ \int_0^t \lim_{T \rightarrow + \infty}  \left( \frac{||\Gamma_s(T)||^2}{2(T-s)}  -(1-\alpha)
\frac{||\Gamma^{\perp}_s(T)||^2}{(T-s)}  \right) ds\\
l_t&=&l_0+ \int_0^t \lim_{T \rightarrow + \infty}  \left( \frac{(2 \alpha -1)||\Gamma^{\perp}_s(T)||^2}{2(T-s)}  +
\frac{||\Gamma^{\sigR}_s(T)||^2}{2(T-s)}  \right) ds
  \end{eqnarray*}
Thus if  $\lim_{T \rightarrow + \infty}
 \frac{||\Gamma_t(T)||^2}{T-t}>0$,  $l_t$ is an non-decreasing function of the risk aversion $\alpha$: if $\alpha \geq  \frac{1}{2}$, 
$l_t$ is  a non-decreasing process as in the forward case; 
 if  $\alpha<1/2$, $l_t$  may be decreasing or increasing, depending on the sign of $\lim_{T \rightarrow + \infty} (\frac{(2 \alpha
-1)||\Gamma^{\perp}_s(T)||^2}{2(T-s)}  + \frac{||\Gamma^{\sigR}_s(T)||^2}{2(T-s)})$. In particular, $\lim_{T \rightarrow + \infty}
 \frac{||\Gamma_t(T)||^2}{T-t}>0$, $\alpha < \frac{1}{2}$ and $\lim_{T \rightarrow + \infty}
\frac{||\Gamma^{\sigR}_t(T)||^2}{2(T-t)}=0$ implies a decreasing yield curve for infinite maturity in this framework of backward power utilities in log-normal market.\\

An affine factor model makes it possible to extend the
previous log-normal model to a more stochastic framework, while leading to tractable pricing
formulas, see \cite{MrNek04}.\\

 {\bf Conclusion}: In this paper, we  remained deliberately closed to the economic setting, studying more precisely power utility functions
and using marginal utility indifference price (Davis price) for the pricing of non replicable zero-coupon bonds, which allowed us to interpret the Ramsey rule in 
a financial framework. Those simplifications imply that  the impact of the  initial economic wealth is avoided : on the one hand 
 power utilities imply that the optimal processes are linear with respect to the initial conditions,  on the other hand  Davis price is a linear 
pricing rule while for non replicable claim, the size of the transaction is an important source of risk that must be taken into account.
 This important issue concerning  the dependency on the initial wealth  and its impact for yield curves will be discussed
in a future work.
\bibliographystyle{plain}
\bibliography{UtilityConsumption}

\begin{thebibliography}{10}

\bibitem{Mike02}
F.P. Berrier and M.R. Tehranchi.
\newblock Forward utility of investment and consumption.
\newblock {\em Preprint}, 2011.

\bibitem{bjork}
T.~Bjork.
\newblock Equilibrium theory in continuous time.
\newblock 2012.

\bibitem{CIR}
John~C. Cox, Jonathan~C. Ingersoll, and Steven~A. Ross.
\newblock A theory of the term structure of interest rates.
\newblock {\em Econometrica}, 53(2):385--403, March 1985.

\bibitem{Davis}
Mark~H.A. Davis.
\newblock Option pricing in incomplete markets.
\newblock In S.R. Pliska, editor, {\em Mathematics of Derivative Securities},
  pages 216--226. M.A.H. Dempster and S.R. Pliska, cambridge university press
  edition, 1998.

\bibitem{Platen}
E.Platen and D.Heath.
\newblock {\em A benchmark approach to quantitative finance}.
\newblock Springer Finance. Springer-Verlag, Berlin, 2006.

\bibitem{FilPlat2009}
D.~Filipovic and E.~Platen.
\newblock Consistent market extensions under the benchmark approach.
\newblock {\em Mathematical Finance}, 19(1):41--52, 2009.

\bibitem{Gollier16}
Christian Gollier.
\newblock An evaluation of sten's report on the economics of climate change.
\newblock Technical Report 464, IDEI Working Paper, 2006.

\bibitem{Gollier13}
Christian Gollier.
\newblock Comment int{\'e}grer le risque dans le calcul {\'e}conomique?
\newblock {\em Revue d'{\'e}conomie politique}, 117(2):209--223, 2007.

\bibitem{Gollier6}
Christian Gollier.
\newblock The consumption-based determinants of the term structure of discount
  rates.
\newblock {\em Mathematics and Financial Economics}, 1(2):81--101, July 2007.

\bibitem{GollierEcological}
Christian Gollier.
\newblock Ecological discounting.
\newblock IDEI Working Papers 524, Institut d'{\'E}conomie Industrielle (IDEI),
  Toulouse, July 2009.

\bibitem{Gollier15}
Christian Gollier.
\newblock Expected net present value, expected net future value and the ramsey
  rule.
\newblock Technical Report 557, IDEI Working Paper, June 2009.

\bibitem{Gollier14}
Christian Gollier.
\newblock Managing long-term risks.
\newblock 2009.

\bibitem{Gollier9}
Christian Gollier.
\newblock Should we discount the far-distant future at its lowest possible
  rate?
\newblock {\em Economics: the Open Access, Open-Assessment E-Journal},
  3(2009-25), June 2009.

\bibitem{HJM}
David Heath, Robert Jarrow, and Andrew Morton.
\newblock Bond pricing and the term structure of interest rates: A new
  methodology for contingent claims valuation.
\newblock {\em Econometrica}, 60(1):77--105, January 1998.

\bibitem{Kunita:01}
H.Kunita.
\newblock {\em Stochastic flows and stochastic differential equations},
  volume~24 of {\em Cambridge Studies in Advanced Mathematics}.
\newblock Cambridge University Press, Cambridge, 1997.
\newblock Reprint of the 1990 original.

\bibitem{Karatzas04}
S.E.Shreve I.~Karatzas, J.P.Lehoczky and G.L.Xu.
\newblock Martingale and duality methods for utility maximization in an
  incomplete market.
\newblock {\em SIAM J. Control Optim.}, 29(3):702--730, 1991.

\bibitem{KaratzasShreve:01}
I.Karatzas and S.E.Shreve.
\newblock {\em Methods of Mathematical Finance}.
\newblock Springer, September 2001.

\bibitem{ElKaroui}
N.El Karoui.
\newblock Les aspects probabilistes du contr\^ole stochastique.
\newblock In {\em Ninth {S}aint {F}lour {P}robability {S}ummer {S}chool---1979
  ({S}aint {F}lour, 1979)}, volume 876 of {\em Lecture Notes in Math.}, pages
  73--238. Springer, Berlin, 1981.

\bibitem{MrNek03}
N.El Karoui and M.~Mrad.
\newblock Mixture of consistent stochastic utilities, and a priori randomness.
\newblock {\em preprint.}, 2010.

\bibitem{MrNek02}
N.El Karoui and M.~Mrad.
\newblock Stochastic utilities with a given optimal portfolio : approach by
  stochastic flows.
\newblock {\em Preprint.}, 2010.

\bibitem{MrNek01}
N.El Karoui and M.~Mrad.
\newblock An exact connection between two solvable sdes and a non linear
  utility stochastic pdes.
\newblock {\em SIAM Journal on Financial Mathematics,}, 4(1):697--736, 2013.

\bibitem{MrNek04}
N.El Karoui, M.~Mrad, and C.~Hillairet.
\newblock Ramsey rule with progressive utility \\ and long term affine yield
  curves.
\newblock {\em To appear in Journal of Financial Engineering.}, 2014.

\bibitem{Kramkov4}
Dmitry Kramkov and Walter Schachermayer.
\newblock Necessary and sufficient conditions in the problem of optimal
  investment in incomplete markets.
\newblock {\em Annals of Applied Probability}, 13(4):1504--1516, 2003.

\bibitem{Hourcade}
Franck Lecocq and Jean-Charles Hourcade.
\newblock Le taux d'actualisation contre le principe de pr\'ecaution?
\newblock Le\c{c}ons \`a partir du cas des politiques climatiques.

\bibitem{JYC}
M.~Chesney M.~Jeanblanc, M.~Yor.
\newblock {\em Mathematical Methods for Financial Markets}.
\newblock Springer Finance, Springer-Verlag, Berlin, 2004.

\bibitem{Mania}
M.Mania and R.Tevzadze.
\newblock Backward stochastic pdes related to the utility maximization problem.
\newblock {\em Georg. Math.J.}

\bibitem{zar-07}
M.Musiela and T.Zariphopoulou.
\newblock Stochastic partial differential equations in portfolio choice.
\newblock {\em Preliminary report}, 2007.

\bibitem{zar-08}
M.Musiela and T.Zariphopoulou.
\newblock Portfolio choice under dynamic investment performance criteria.
\newblock {\em Quantitative Finance}, 9(2):161--170, 2009.

\bibitem{zar-07a}
M.~Musiela and T.~Zariphopoulou.
\newblock Investment and valuation under backward and forward dynamic
  exponential utilities in a stochastic factor model.
\newblock In {\em Advances in mathematical finance}, pages 303--334.
  Birkh\"auser Boston, 2007.

\bibitem{KarEng}
N.Englezos and I.Karatzas.
\newblock Utility maximization with habit formation: Dynamic programming and
  stochastic pdes.
\newblock {\em SIAM J. Control Optim.}, 48(2):481--520, 2009.

\bibitem{ElKarouiFrachot}
Antoine~Frachot Nicole~ElKaroui and Helyette Geman.
\newblock On the behavior of long zero coupon rates in a no arbitrage
  framework.
\newblock {\em Review of derivatives research}, 1:351--369, 1997.

\bibitem{Protter}
P.E.Protter.
\newblock {\em Stochastic integration and differential equations}, volume~21 of
  {\em Stochastic Modelling and Applied Probability}.
\newblock Springer-Verlag, Berlin, 2005.
\newblock Second edition. Version 2.1, Corrected third printing.

\bibitem{Dybvig}
J.E.~Ingersol P.H.~Dybvig and S.A. Ross.
\newblock Longfforward and zero-coupon rates can never fall.
\newblock {\em Journal of Business}, 69:1--25, 1996.

\bibitem{Carmona}
R.A.Carmona and D.Nualart.
\newblock {\em Nonlinear stochastic integrators, equations and flows}, volume~6
  of {\em Stochastics Monographs}.
\newblock Gordon and Breach Science Publishers, New York, 1990.

\bibitem{Ramsey}
F.P. Ramsey.
\newblock A mathematical theory of savings.
\newblock {\em The Economic Journal}, (38):543--559, 1928.

\bibitem{Rogers}
L.C.G. Rogers.
\newblock A mathematical theory of savingsduality in constrained optimal
  investment and consumption problems: A synthesis.
\newblock {\em Working paper, Statistical Laboratory, Cambridge
  University.<http://www.statslab.cam.ac.uk/~chris/>}, 2003.

\bibitem{MR2014244}
W.~Schachermayer.
\newblock A super-martingale property of the optimal portfolio process.
\newblock {\em Finance Stoch.}, 7(4):433--456, 2003.

\bibitem{Ventzel}
A.D. Ventzel.
\newblock On equations of the theory of conditional markov.

\bibitem{Weitzman}
Martin~L. Weitzman.
\newblock Why the far-distant future should be discounted at its lowest
  possible rate.
\newblock {\em Journal of Environmental Economics and Management},
  36(3):201--208, November 1998.

\bibitem{Weitzman_review}
Martin~L. Weitzman.
\newblock A review of the the stern review on the economics of climate change.
\newblock {\em Journal of Economic Litterature}, 45:703--724, September 2007.

\end{thebibliography}

\end{document}